\documentclass[aps,pre,floats,floatfix,superscriptaddress,nofootinbib,twocolumn,showkeys,longbibliography]{revtex4-2}
\usepackage{bbm}
\usepackage{amsmath,amssymb,amsthm}
\usepackage{xcolor}
\usepackage{graphicx}
\usepackage{caption}
\usepackage{makecell}
\usepackage{url}
\usepackage{answers}
\usepackage{array}
\usepackage[mathscr]{euscript}
\usepackage{eufrak}
\usepackage{calligra}
\usepackage{algpseudocode,algorithm}
\usepackage{subcaption}

\usepackage{amsthm}

\newtheorem{theorem}{Theorem}
\newtheorem{lemma}[theorem]{Lemma}

\usepackage[colorlinks=true, linkcolor=blue, citecolor=blue, urlcolor=blue]{hyperref}
\usepackage{dcolumn}% Align table columns on decimal point
\usepackage{bm}% bold math
\usepackage{booktabs}
\newcommand{\ZZZ}[2]{}

\newcommand{\quotes}[1]{``#1''}
\newtheorem{property}{Property}
\newtheorem{problem}{Problem}

\begin{document}

\preprint{APS/123-QED}

\title{\textbf{Flow Subgraphs and Flow Network Design under End-to-End Power Dissipation Constraints} 
}% 

% \author{Zhihao~Qiu}
%  % \altaffiliation[Also at ]{Physics Department, XYZ University.}%Lines break automatically or can be forced with \\
% \author{S\'amuel G. Balogh}%
% \affiliation{the Faculty of Electrical Engineering, Mathematics and Computer Science, Delft University of Technology, 2600 GA Delft, The Netherlands.
% }%

% \author{Xinhan~Liu}
% \affiliation{
%  the Faculty of Electrical Engineering, Mathematics and Computer Science, Delft University of Technology, 2600 GA Delft, The Netherlands.
% }%
% \affiliation{
%  second institution for this author
% }%
% \author{Piet Van Mieghem}
% \affiliation{%
%  Authors' institution and/or address\\
%  This line break forced with \textbackslash\textbackslash
% }%

% \author{Maksim~Kitsak}
% \email{Contact author: m.a.kitsak@tudelft.nl}
% \affiliation{%
%  Authors' institution and/or address\\
%  % This line break forced with \textbackslash\textbackslash
% }%

\author{Zhihao~Qiu}
\affiliation{
Faculty of Electrical Engineering, Mathematics and Computer Science,
Delft University of Technology,
2600 GA Delft, The Netherlands}

\author{Xinhan~Liu}
\email{Contact author: X.Liu-22@tudelft.nl}
\affiliation{
Faculty of Electrical Engineering, Mathematics and Computer Science, 
Delft University of Technology, 
2600 GA Delft, The Netherlands}

\author{Rogier Noldus}
\affiliation{
Faculty of Electrical Engineering, Mathematics and Computer Science, 
Delft University of Technology, 
2600 GA Delft, The Netherlands}
\affiliation{Ericsson, The Netherlands}

\author{Piet~Van~Mieghem}
\affiliation{
Faculty of Electrical Engineering, Mathematics and Computer Science, 
Delft University of Technology, 
2600 GA Delft, The Netherlands}

% \collaboration{CLEO Collaboration}%\noaffiliation

\date{\today}% It is always \today, today,
             %  but any date may be explicitly specified

\begin{abstract}
We investigate how the underlying graph of a network supports a flow between a source node and a destination node and propose to compute the expected number of nodes and links that contribute to transferring items in random graphs.   
Since the transportation is associated with a \quotes{cost} or \quotes{power dissipation}, we further address how to construct a graph given predetermined end-to-end power dissipation, which can be reduced to the \quotes{inverse effective resistance problem} that asks for a weighted graph in which the effective resistance matrix equals a predetermined demand matrix.
We propose a heuristic algorithm, \quotes{Resistor Gap Pruning} (RGP), which provides sparse graphs closely approximating the demand effective resistance and which shows stable performance across different demand scenarios.
\end{abstract}

%Use showkeys class option if keyword
                              %display desired
\maketitle

\section{Introduction}\label{sec:Introduction}
A network is characterized by the underlying graph and the functional process of the network~\cite{newman2003structure,van2014performance}.
A graph $G$ consisting of a set $\mathcal{N}$ of $N$ nodes that are connected by a set $\mathcal{L}$ of $L$ links and the function of a network is generally related to the transport of items over the underlying graph~\cite{van2010framework,PVM_GraphSpectra2023,qiu2023inverse}.
Depending on the type of transported items, a network can be classified as either a \quotes{path network} in which packets (e.g., IP packets, vehicles) are propagated, or a \quotes{flow network}, e.g., electrical networks, gas networks, water networks, etc, in which the transported item constitutes a flow.

In a path network, the transport of items between two nodes $(i,j)$ travels along a single path $\mathcal{P}_{ij}$, defined as a succession of links \cite{PVM_GraphSpectra2023,van2014performance} of the form $\mathcal{P}_{ij} = \{n_1\sim n_2, n_2\sim n_3, \dots, n_{k-1} \sim n_k\}$, where $k-1$ is the hop count of the path and where node label $n_1=i$, $n_k=j$ and $n_a \neq n_b$ for each index $a$ and $b$ of a link $n_a \sim n_b$ in $\mathcal{P}_{ij}$. 
In most practical scenarios, the transport follows the \quotes{the shortest path} $\mathcal{P}^*_{ij}$, which minimizes the sum of all link weights along the path.
A well-known example is the Open Shortest Path First (OSPF) protocol~\cite{fortz2000internet, van2014performance} in data communication networks, where routers in an IP network maintain a list of IP destinations and the \quotes{next hop} toward each destination. 
The next hop is the link selected by the router to route an IP packet over the shortest path towards the destination.
The shortest path problem has been extensively studied in literature~\cite{van2014performance, dijkstra1959note,bellman1958routing,schrijver2012history,cormen2022introduction,kitsak2023finding,gomathi2018energy, zhang2024mapreduce,zhan1998shortest,begtavseviu2001measurements,qiu2022efficient}. 

Many practical scenarios, however, cannot be adequately described by the shortest path alone~\cite{akara-pipattana_resistance_2022,NEWMAN200539,freeman_centrality_1991,BOZZO2013460,barabasi2013network}.
Multiple alternative paths, or even random paths, are considered when the shortest paths are unknown or unavailable.
For example, next-generation 6G communication systems aim to enable information to be divided into smaller units.
These smaller units are transmitted simultaneously over diverse paths, including satellite, fiber-optic and wireless links, to maximize coverage, reliability and transmission efficiency~\cite{s22093136,8766143}.
Other examples include the epidemic spreading~\cite{van2014performance,pastor2015epidemic,nowzari2016analysis}, where infections can propagate through all available connections.
Similarly, the propagation of news or rumors in a social network generally does not follow the shortest path from a source node to a destination node, but rather resembles a random walk process~\cite{Bozzo2012Approximations,NEWMAN200539,MASUDA20171,pearson1905problem}.
If the shortest path is not sufficient for information transfer, then flow networks provide a complementary perspective, since the transport propagates over all possible paths from node $i$ to node $j$ in a flow network.

In this work, we focus on the flow network using the current-voltage analogy.
Items are transferred following the principle of current flow through an electrical network, where links are resistors and nodes are junctions between the resistors.
We first investigate how the flow transmits through a flow network.
Specifically, we study how the underlying graph of a network supports a flow between the source node and the destination node and which nodes and links are included in the transmission.
We refer to the nodes and links that contribute to transferring flow as the \quotes{flow subgraph}.
We propose a method to compute the expected number of nodes and links in the flow subgraph for Erd\H{o}s--R\'enyi (ER) random graphs.   

The flow over nodes and links in an underlying graph also affects the \quotes{power dissipation} in the links of the graph. 
The investigation of the power dissipation is inspired by communication networks, where higher data transmission through a link generally leads to higher cost for that connection~\cite{guennebaud2024energy}; the higher cost may be installation cost, e.g., fibre optic cable, as well as operational cost, e.g., energy consumption of routers or repeaters.
Similar examples can be observed in social networks~\cite{holme2009diplomat,jackson2007study} and transportation networks~\cite{hassidim2013network,aldous2008optimal}. In this work, we address the challenge of constructing a flow network whose total power dissipation, induced by flow transmission between various source–destination node pairs, approximately matches the predetermined demands. 
% We demonstrate that transferring items in accordance with the \quotes{resistance} of each path between the source node and the destination node provides a more evenly spread load on the links in the network and may lead to a reduction in power consumption compared to the transmission following one single path.
% The power dissipated on each link in ER graphs holds a power law decay with respect to the size of the graph and link density. 

% Finally, we address the challenge of constructing a flow network whose total power dissipation, induced by flow transmission between various source–destination node pairs, approximately matches the predetermined demands.

The outline of the paper is as follows:
In Section~\ref{sec:Notation}, we introduce terminology and background knowledge on flow networks.
In Section~\ref{sec:Sizeoftheflowsubgraph}, we investigate the links and nodes that contribute to information/data transmission in flow networks.
A method is proposed to compute the expected number of nodes and links that contribute to transferring items in random graphs.
% We then present results about the power dissipation on both network level and link level in Section~\ref{sec: power dissipation in flow networks}.
In Section~\ref{sec:Network construction with end-to-end demands on power dissipation}, we consider an inverse problem of power dissipation, which asks for constructing a flow network given predetermined demands on the end-to-end power dissipation. 
A heuristic approach named Resistor Gap Pruning (RGP) is proposed and tested, given different demand matrices.
Finally, we summarize our results in Section~\ref{sec:Conclusion}.

\section{Terminology}\label{sec:Notation}
% Complex networks have two general features: a graph defining the underlying topology and a service or function specified by a dynamic process \cite{van2010framework}. 
In this Section, we introduce the terminology and refer to Appendix~\ref{app:Notation list} for the symbol list.

\subsection{Graph}
\label{sec:graph}
We limit ourselves to connected, undirected, simple graphs\footnote{A simple graph has no multiple links between a same pair of nodes and the graph does not contain any self-loops, i.e. $a_{ii}=0$ for each node $i\in\mathcal{N}$.} \cite{PVM_GraphSpectra2023}.
We introduce an $N \times N$ adjacency matrix $A$ to represent a graph, with element ${a}_{ij}=1$ if there is a link between node $i$ and node $j$, otherwise ${a}_{ij}=0$.
The link between nodes $i$ and $j$ is represented by $l = i\sim j$.
Each link $l\in \mathcal{L}$ has a positive weight $w_l$ and the $N\times N$ link weight matrix is denoted by $W$. 
The $N\times N$ weighted adjacency matrix is defined as $\widetilde{A}=W \circ A $, where the Hadamard product $\circ$ defines a direct elementwise multiplication  $\widetilde{a}_{ij}=w_{ij}a_{ij}$ and the \quotes{tilde} refers to a weighted graph matrix.
In our setting, $\widetilde{a}_{ij}=0$ means that there is no link between nodes $i$ and $j$, because we exclude zero link weights, i.e. $w_{ij} > 0$, in order to avoid the complication that a zero weight $w_{ij} = 0$ would physically mean that nodes $i$ and $j$ are the same.

\begin{figure}[!htbp]
  \centering
  \includegraphics[scale=0.25]{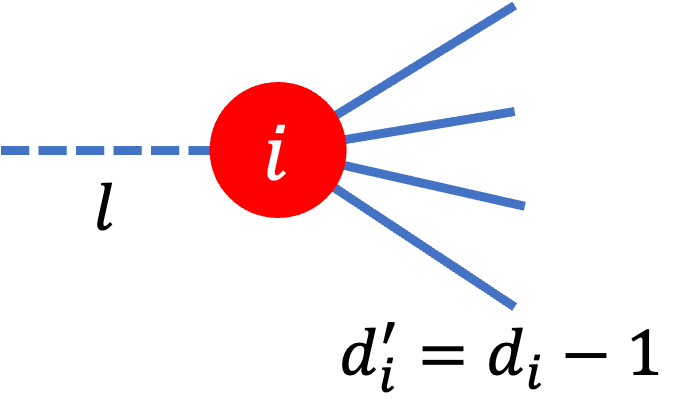}
  \caption{Illustration of the excess degree $d_i'$. When node $i$ is reached by following a randomly chosen link $l$, the remaining number of neighbors of node $i$ is $d_i' = d_i - 1$.
}
\label{fig:excess_degree}
\end{figure}

The degree $d_i$ of node $i$ equals the row sum of the adjacency matrix $A$, i.e. $d_i = \sum_{k=1}^N a_{ik}$, which represents the number of neighbors of node $i$.
The excess degree $d^\prime_i$ of node $i$ is defined~\cite{newman2010networks, PhysRevE.64.026118} as the number of remaining links after reaching node $i$ by a random link $l$.
In other words, for node $i$ with degree $d_i$, the excess degree $d^\prime_i = d_i-1$, as shown in Fig.~\ref{fig:excess_degree}.
Suppose a graph has degree distribution $\Pr[D=k]$.
The degree probability generating function (pgf) is defined as
\[
\varphi_D(z)=E[z^D]=\sum_{k=0}^{N-1}\Pr[D=k]\,z^k,
\]
which encodes the full distribution of the degree $D$ into an analytic function~\cite[Chapter 2]{van2014performance}.
The excess degree distribution~\cite{newman2010networks} is then 
\begin{equation}
    \text{Pr}\left[D^\prime = k\right] = \frac{k+1}{E[D]}\text{Pr}\left[D = k+1\right],
    \notag
\end{equation}
where $E[D]$ is the expected degree of the graph.
The pgf of the distribution of excess degree $D'$ is
\begin{equation}
    \varphi_{D^\prime}(z) = E\left[ z^{D^\prime}\right] = \sum_{k=0}^{N-1}\Pr[D^\prime=k]z^k = \frac{\varphi^\prime_{D}(z)}{\varphi_{D}^\prime(1)}
    \notag
\end{equation}

A directed graph can also be represented by an $N\times L$ incidence matrix $B$ with elements
\begin{equation}
    b_{il}=
    \left\{
    \begin{array}{rcl}
    1       &      & \text{if link $l = i \rightarrow j$}\\
    -1     &      & \text{if link $l = j \rightarrow i$}\\
    0     &      & \text{otherwise}
    \end{array} \right.
    \notag
\end{equation}
where link $l = i \rightarrow j$ denotes the direction from node $i$ to node $j$.
The sum of the columns in an incidence matrix equals zero, i.e. $u^TB=0$, where $u$ represents the $N\times 1$ all-one vector. 
An undirected graph can be represented by an $N\times (2L)$ incidence matrix, where each link $i\sim j$ is counted twice, once for direction $i \rightarrow j$ and once for direction $j \rightarrow i$. In that case, the degree of each node is just doubled~\cite{PVM_GraphSpectra2023}. 
% The $N\times 1$ degree vector $d=A u$ contains the degree of each node, with $d_i$ denoting the node $i$ degree and $u$ representing the $N\times 1$ all-one vector. 
% The $N\times N$ degree diagonal matrix $\Delta = \text{diag}(d)$ contains the node degrees on its main diagonal.

% A graph $G\left(\mathcal{N},\mathcal{L}\right)$ consists of a set $\mathcal{N}$ of $N=|\mathcal{N}|$ nodes, interconnected via a set $\mathcal{L}$ of $L=|\mathcal{L}|$ links. Node-pair connections are captured by the $N\times N$ adjacency matrix $A$, where $a_{ij}=1$ if node $i$ and node $j$ are connected, otherwise $a_{ij}=0$. In this work, we focus on connected, simple\footnote{A simple graph is a graph without self-loops and multiple links.} graphs. 
% The $N\times 1$ degree vector $d=A  u$ contains the degree of each node, with $d_i$ denoting the node $i$ degree, where the $N\times 1$ all-one vector is denoted as $u$. The $N\times N$ degree diagonal matrix $\Delta = \text{diag}(d)$ contains node degrees on its main diagonal.

\subsection{Laplacian matrix Q}\label{sec:Laplacian}
For an undirected graph, the $N\times N$ Laplacian matrix $Q$ reveals the relation~\cite{PVM_GraphSpectra2023} between the adjacency matrix $A$ and the $N\times L$ incidence matrix $B$:
\begin{equation}
    \label{eq:unweightedLaplacian}
    Q = BB^T = \Delta - A
\end{equation}
where the $N\times N$ degree diagonal matrix $\Delta = \text{diag}(d)$ contains the node degrees on its main diagonal and the $N\times 1$ degree vector $d=A  u$ contains the degree $d_i$ of node $i$ as a component.
The eigenvalue decomposition \cite{van2017pseudoinverse} of the Laplacian $Q$ is
\begin{equation}\label{Eq:Q_Spectrum}
	Q = Z  \text{diag}(\mu)   Z^T,
\end{equation}
defining the set of $N$ orthogonal $N\times 1$ eigenvectors $z_i$ of $Q$, contained in columns of the $N\times N$ eigenvector matrix $K$ and $N$ eigenvalues $\mu_1 \geq \mu_2 \geq \dots \geq \mu_N = 0$ of $Q$.
The Laplacian is not invertible, because $\det Q = \prod_{k=1}^{N} \mu_k =0$ due to $\mu_N = 0$.
The pseudoinverse of the Laplacian \cite{van2017pseudoinverse}
\begin{equation}\label{Eq:Q_pinv}
	Q^{\dagger} = \sum_{k=1}^{N-1} \frac{1}{\mu_k}z_k z_k^T
\end{equation}
% Due to double orthogonality of the eigenvector matrix $K$ (i.e. $Z  Z^T = I$ and $Z^T  Z = I$), where $I$ is the $N\times N$ identity matrix, relation (\ref{Eq:Q_Spectrum}) can be transformed into a weighted sum of $N$ outer vector products
% \begin{equation}\label{Eq:Q_Spectrum_Sum}
% 	Q = \sum\limits_{i=1}^{N} \mu_i   z_i   z_i^T.
% \end{equation}
% As of any real, symmetric matrix \cite{PVM2010GraphSpectra}, the eigenvalues of Laplacian $Q$ are real and non-negative because $Q$ is a positive semidefinite matrix \cite[p.67]{PVM2010GraphSpectra}. 
% From $Q  u = 0$, we observe that $\mu_N = 0$ and $z_N = u$ and thus $\det Q = 0$. Consequently, the Laplacian $Q$ is not invertible. However, the pseudoinverse\footnote{We restrict the analysis to connected graphs, as the number of zero eigenvalues of Laplacian $Q$ equals the number of connected components in a graph. More precisely, relation (\ref{Eq:Q_pinv}) does not hold in the case of a disconnected graph.} 
% \begin{equation}\label{Eq:Q_pinv}
% 	Q^{\dagger} = \sum\limits_{i=1}^{N-1}\frac{1}{\mu_i}  z_i   z_i^T
% \end{equation}
obeys $Q^{\dagger}  Q = Q  Q^{\dagger} = I - \frac{1}{N}  J$, where $I$ is the identity matrix and $J$ is the $N\times N$ all-one matrix.
The weighted Laplacian matrix is defined as $\widetilde{Q}=\widetilde{\Delta}-\widetilde{A}$, where the $N\times N$ weighted degree diagonal matrix $\widetilde{\Delta} = \text{diag}(\widetilde{d})$ contains weighted node degrees on its main diagonal and the $N\times 1$ weighted degree vector $\widetilde{d}=\widetilde{A} u$ contains the weighted degree of each node.
Details of the Laplacian and pseudoinverse of the Laplacian are provided in \cite{PVM_GraphSpectra2023} and \cite{van2017pseudoinverse}.

\subsection{Electrical resistor network and flow subgraph}
\label{sec:Electrical resistor network}
In this work, the flow network is modeled analogously to a resistor network.
Consider a resistor network $G$ with $N$ nodes and $L$ links.
Each link $l = i \sim j$ between the nodes $i$ and $j$ possesses a resistor with resistance $r_l = r_{ij}>0$ and the weight of link $w_{l} = w_{ij} = \frac{1}{r_l}$.
Let $v_k$ denote the potential or voltage of node $k$ in the network and the current $y_l=y_{ij}$ through the resistor of link $l$ between node $i$ and node $j$.
The link current is directed so that $y_{ij} = -y_{ji}$. 
Ohm\textquotesingle s law defines the relation $v_i-v_j = r_{ij} y_{ij}$, which states that the resistance is the proportionality constant or ratio between the potential differences $v_i-v_j$ and the current $y_{ij}$.
The voltage vector and link current vector are related by
\begin{equation}\label{Eq:Ohm law}
	y = \text{diag}\left(\frac{1}{r_l}\right) B^Tv,
\end{equation}
where $y$ is the $L\times1$ link current vector, $v$ is the $N\times 1$ vector with nodal voltages and the $L\times L$ matrix $\text{diag}(\frac{1}{r_l})$ has diagonal elements $(\frac{1}{r_1},\dots,\frac{1}{r_2},\dots,\frac{1}{r_L})$ where $r_l = r_{ij}$ is the resistance of link $l = i \sim j$.
If the current through link $l = i \sim j$ is zero, then
\begin{equation}
\label{eq:yl=0}
    y_l = 0
\quad\Longleftrightarrow\quad
v_i = v_j
\end{equation}

Define $x_i$ as the external current injected into node $i$.
If the current leaves node $i$, then $x_i<0$, otherwise $x_i>0$.
Kirchhoff’s current law shows that
\begin{equation}
    x = By,
    \label{Eq:Kirchhoff law}
\end{equation}
where the $N\times1$ node external current vector $x$ contains the external current injected into each node $x_i$.
Substituting (\ref{Eq:Ohm law}) into Kirchhoff’s current law (\ref{Eq:Kirchhoff law}) relates the external current and voltage,
\begin{equation}
    x = B\text{diag}\left(\frac{1}{r_{ij}}\right)B^Tv = \widetilde{Q}v
    \label{Eq:external current vs voltage}
\end{equation}
where the weighted Laplacian
\begin{equation}
    \widetilde{Q} = B\text{diag}\left(\frac{1}{r_{ij}}\right)B^T
\end{equation}
is defined similarly to the unweighted Laplacian decomposition in \eqref{eq:unweightedLaplacian}.

Although matrix relation $x = \widetilde{Q}v$ in (\ref{Eq:external current vs voltage}) cannot be inverted due to $\text{det}~\widetilde{Q} =0$, Van Mieghem~\cite{PVM_GraphSpectra2023} shows that   
\begin{equation}
    v = \widetilde{Q}^\dagger x,
    \label{Eq:voltage vs external current}
\end{equation}
by choosing the average voltage $\frac{1}{N}u^Tv$ equal to zero.
Substituting the voltage $v$ in (\ref{Eq:Ohm law}) with (\ref{Eq:voltage vs external current}) couples the external current vector and link current vector:
\begin{equation}\label{Eq:external current vs inject current}
	y = Cx,
\end{equation}
where the $L\times N$ matrix $C = \text{diag}\left(\frac{1}{r_l}\right) B^T \widetilde{Q}^\dagger$.

If we inject a unit current $I_c = 1$ Ampere at the source node $i$, which leaves the network $G$ at the destination node $j$, then the link current
\begin{equation}\label{Eq:yvsxgivenst}
    y = C(e_i-e_j),
\end{equation}
where $e_k$ represents an $N\times 1$ basic vector that has only one non-zero element $(e_k)_k = 1$.
For a link $l = m \sim n$, \eqref{Eq:yvsxgivenst} and $C = \text{diag}\left(\frac{1}{r_l}\right) B^T \widetilde{Q}^\dagger$ lead to
\begin{align}
\label{Eq:ylvsxgivenst}
    y_l & = \frac{1}{r_l}(e_m-e_n)^{T}\widetilde{Q}^\dagger(e_i-e_j) \notag\\
    & = \frac{1}{r_l} \left(\widetilde{Q}^\dagger_{mi} - \widetilde{Q}^\dagger_{mj} - \widetilde{Q}^\dagger_{ni} +\widetilde{Q}^\dagger_{nj}\right)
\end{align}

For a given source-destination node pair $(i,j)$, the flow subgraph $G^*_{ij}$ is defined as a subgraph of the resistor network $G$ through which current propagates.
Specifically, the flow subgraph $G^*_{ij}$ consists of all links carrying nonzero current, along with the nodes incident to these links. 
Formally, the link set $\mathcal L(G^*_{ij})$ of the flow subgraph $G^*_{ij}$ is
\[
\mathcal L(G^*_{ij})
:= \{\, l \in \mathcal L(G) \,|\, y_{l}\neq 0 \,\},
\]
and the node set $\mathcal N(G^*_{ij})$ is
\begin{align}
    \mathcal N(G^*_{ij}) \notag
&:= \{\, m\in\mathcal N(G) \,|\, \exists\, n\in\mathcal N(G)\\ &\text{ such that } l = m\sim n \in
\mathcal L(G^*_{ij}) \,\} \notag
\end{align}

% \mathcal N(G^*_{ij})
% := \{\, m\in\mathcal N(G) \,|\, \exists\, n\in\mathcal N(G) \text{ such that } l = m\sim n \in
% \mathcal L(G^*_{ij}) \,\}
% \]

\subsection{Effective resistance and power dissipation}\label{sec:Omega}
Consider the electrical resistor network $G$ introduced in Section~\ref{sec:Electrical resistor network}, where each link $l = i \sim j$ possesses a resistor with resistance $r_{ij}=\frac{1}{w_l}$.
The ratio $\frac{v_i-v_j}{I_c}$  measures the resistance of a subgraph over which the injected current $I_c$ in node $i$ spreads towards node $j$ and $\omega_{ij}$ is called the \quotes{effective} resistance~\cite{PVM_GraphSpectra2023} between node $i$ and node $j$.
As demonstrated in \cite{PVM_GraphSpectra2023,van2017pseudoinverse}, the effective resistance $\omega_{ij}$ between node $i$ and node $j$ satisfies
\begin{equation}\label{Eq:omega_ij}
\omega_{ij} = \left(e_i - e_j\right)^{T}  \widetilde{Q}^{\dagger} \left(e_i - e_j\right)
\end{equation}
% where the $N\times 1$ basic vector $e_i$ has only one non-zero element $(e_i)_i = 1$.
Multiplying \eqref{Eq:omega_ij} out yields
\begin{equation}\label{Eq:omega_ij2}
\omega_{ij} = \widetilde{Q}^{\dagger}_{ii}+\widetilde{Q}^{\dagger}_{jj} - 2\widetilde{Q}^{\dagger}_{ij}
\end{equation}
Relation (\ref{Eq:omega_ij2}) can be transformed into a matrix form, defining the $N\times N$ effective resistance matrix 
\begin{equation}\label{Eq:Omega}
\Omega = \zeta   u^T + u   \zeta^T - 2  \widetilde{Q}^{\dagger},
\end{equation}
where the $N\times 1$ vector $\zeta = \left(\widetilde{Q}^{\dagger}_{11}, \, \widetilde{Q}^{\dagger}_{22}, \, \dots ,\, \widetilde{Q}^{\dagger}_{NN} \right)$ contains the diagonal elements of the pseudoinverse $\widetilde{Q}^{\dagger}$ of the weighted Laplacian $\widetilde{Q}$.

We express the link current $y_l$ \eqref{Eq:ylvsxgivenst} in terms of the effective resistance (\ref{Eq:omega_ij2}) as:
\begin{align}
\label{Eq:ylvsomega}
    y_l = \frac{1}{2r_l} \left(\omega_{mj}+\omega_{ni}-\omega_{mi} - \omega_{nj}\right)
\end{align}
and \eqref{Eq:ylvsomega} shows that
\begin{equation}
\label{eq:yl=02}
    y_l = 0
\quad\Longleftrightarrow\quad
\omega_{mj}-\omega_{mi} = \omega_{nj}-\omega_{ni}
\end{equation}

The effective resistance $\omega_{ij}$ between two adjacent nodes $i$ and $j$ (i.e. $a_{ij} = 1$), represents the effective resistance of a parallel connection
\begin{equation}\label{Eq:Omega_ij_parallel}
    \frac{1}{\omega_{ij}} = \frac{1}{r_{ij}} + \frac{1}{\left(\omega_{G\setminus l}\right)_{ij}},
\end{equation}
where $\left(\omega_{G\setminus l}\right)_{ij}$ is the effective resistance between node $i$ and $j$ in the graph $G\setminus l$ obtained from the graph $G$ after deletion of the link $l= i \sim j$.
The effective resistance $\omega_{ij}$ between adjacent nodes $i$ and $j$ in \eqref{Eq:Omega_ij_parallel} is upper bounded by the resistance $r_{ij}$ of the direct link between them
\begin{equation}\label{Eq:Omega_ij_modified}
    \omega_{ij} = \frac{r_{ij} \left(\omega_{G\setminus l}\right)_{ij}}{r_{ij} + \left(\omega_{G\setminus l}\right)_{ij}} \leq \min \left(r_{ij}, \, \left(\omega_{G\setminus l}\right)_{ij} \right)
\end{equation}
% Otherwise, if $a_{ij}=0$, then the effective resistance $\omega_{ij}$ is upper bounded by the sum of resistances of links forming the shortest path between the nodes. 

Inject a current $I_c$ in a network from source node $i$ and let it out from destination node $j$. 
The power dissipation $P_G$, i.e., the energy per unit time (in Watt), in the network is the sum of the power dissipated on each link,
\begin{equation}
    \label{eq:graphpowerdissipation}
    P_G = \sum_{l \in \mathcal{L}} P_l = \sum_{l \in \mathcal{L}} y_l^2r_l,
\end{equation}
which also equals the product of the square of the current and the effective resistance
\begin{equation}
    \label{eq:graphpowerdissipation2}
    P_G = I_c^2 \omega_{ij}
\end{equation}

\section{Size of the flow subgraph}
\label{sec:Sizeoftheflowsubgraph}
Given a flow network $G$ and a pair of nodes $(i,j)$, the transported items propagate only along the corresponding flow subgraph $G^*_{ij}$.  Indeed, the number of nodes and links in a flow subgraph $G^*_{ij}$ is generally determined by the underlying graph $G$.
For example, for a tree graph, the flow subgraph $G^*_{ij}$ between nodes $(i, j)$ is the same as the corresponding (shortest) path $\mathcal{P}^*_{ij}$. 
In a $2$-D lattice graph $G$, the flow subgraphs $G^*_{ij}$ are identical to the graph $G$ for every pair of nodes $(i, j)$. 

In this section, we propose a \quotes{self-consistent approach} to compute the expected number of nodes and links in the flow subgraphs $G^*_{ij}$.
The main idea of the self-consistent approach is inspired by the property: 
\begin{property}
    \label{property1}
    A node $k$, different from node $i$ and $j$, in the flow subgraph $G^*_{ij}$ between nodes $(i, j)$ has at least two neighbors belonging to the flow subgraph $G^*_{ij}$.
\end{property}
\begin{proof}
For an arbitrary node $k$ that belongs to the flow subgraph $G^*_{ij}$, the current $I_k$ flowing through node $k$ is not zero, because nodes outside $G^*_{ij}$ do not carry current. Kirchhoff's current law states that the total current flowing into node $k$ equals the total current leaving node $k$.
Hence, node $k$ must have at least two neighbors belonging to the flow subgraph $G^*_{ij}$.
\end{proof}

In graphs with i.i.d. continuous resistances (equivalently, i.i.d. continuous link weights), Property~\ref{property1} is also sufficient for a node (other than $i$ and $j$) to belong to the flow subgraph.
% For \textbf{weighted} graphs with independent continuous resistances, Property~\ref{property1} is not only necessary but also sufficient for a node (other than $i$ and $j$) to belong to the flow subgraph.
This statement follows from two auxiliary lemmas given in Appendix~\ref{app:B}: (i) the probability that two distinct nodes in the flow subgraph have identical potentials is zero, and (ii) a node connected to at least two flow-subgraph nodes with distinct potentials must itself belong to the flow subgraph.
In Fig.~\ref{sfig:weighted_flow_ex}, we present an example to show the flow subgraph of a graph with i.i.d. link weights.

For graphs with identical link weights (link resistances), Property~\ref{property1} remains only a necessary condition for a node (other than $i$ and $j$) to belong to the flow subgraph. 
Local structural symmetries may lead to identical electrical potentials at distinct nodes, so that a node $k$ can be adjacent to multiple flow-subgraph nodes and all incident links of $k$ carry zero current. 
An example is shown by Fig.~\ref{sfig:unweighted_flow_ex}, where equipotential neighbors (node $u$ and $v$) prevent the propagation of current through the intermediate node (node $w$).
Although Property~\ref{property1} is not a sufficient condition, in random graphs with the same link weights, configurations in which exact structural symmetries lead to equipotential nodes are non-generic and occur with small probability~\cite{erdos1963asymmetric}. 
% The contrast between weighted and unweighted graphs, illustrated in Fig.~\ref{fig:flowsubgraphexample}, clarifies why Property~\ref{property1} leads to a self-consistent characterization of the flow subgraph only in weighted networks.

% For \textbf{unweighted} graphs, Property~\ref{property1} remains only a necessary condition. Local structural symmetries may lead to identical electrical potentials at distinct nodes, so that a node can be adjacent to multiple flow-subgraph nodes while all incident links carry zero current. Such situations are illustrated in Fig.~\ref{sfig:unweighted_flow_ex}, where equipotential neighbors prevent the propagation of current through the intermediate node. The contrast between weighted and unweighted graphs, illustrated in Fig.~\ref{fig:flowsubgraphexample}, clarifies why Property~\ref{property1} leads to a self-consistent characterization of the flow subgraph only in weighted networks.

\begin{figure}[!htbp]
    \centering
    % ---------- 子图 (a) ----------
    \begin{subfigure}[t]{0.21\textwidth}
        \includegraphics[width=\linewidth]{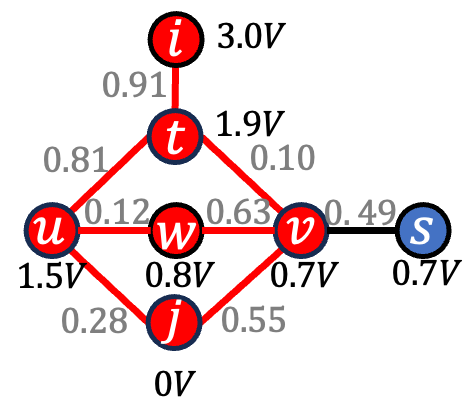}
        \caption{Flow subgraph of a graph with i.i.d. random link weights}
        \label{sfig:weighted_flow_ex}
    \end{subfigure}
    \hspace{10pt}
    % ---------- 子图 (b) ----------
    \begin{subfigure}[t]{0.21\textwidth}
        \includegraphics[width=\linewidth]{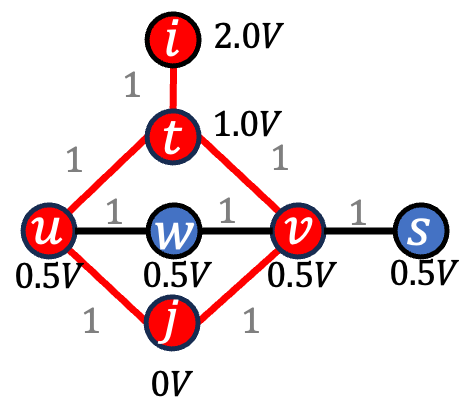}
        \caption{Flow subgraph of a graph with equal link weights}
        \label{sfig:unweighted_flow_ex}
    \end{subfigure}
    % ---------- 总标题 ----------
    \caption{Examples illustrating the role of link-weight degeneracy in the flow subgraph. The nodes and links of the flow subgraph are highlighted in red. In (a), link weights are i.i.d. continuous random variables, which yields distinct nodal potentials almost surely. Thus, Property~\ref{property1} is sufficient. In (b), all link weights are identical and structural symmetries create equipotential nodes $u$ and $v$. Consequently, node $w$ does not belong to the flow subgraph, although it satisfies Property~\ref{property1}.
    % In the weighted case, each link has an independent random resistance, which induces distinct nodal potentials in flow subgraph; consequently, every node satisfying Property~\ref{property1} carries nonzero current and belongs to the flow subgraph. 
    % In the unweighted case, structural symmetries may produce equipotential nodes $u$ and $v$, so that node $w$ even satisfies Property~\ref{property1} is no longer a flow-subgraph node.
    }
    \label{fig:flowsubgraphexample}
\end{figure}

In this section, we analyze the structure of the flow subgraph in Erd\H{o}s--R\'enyi random graphs and compute the expected size of the flow subgraph based on our proposed self-consistent approach.
The self-consistent approach captures the typical structure of flow subgraphs and can be applied to all random graphs with a known probability generating function (pgf) of node degree $D$.

% In this section, we analyze the structure of the flow subgraph in random graphs and propose a self-consistent approach to compute the expected size of the flow subgraph. The approach relies on Property~\ref{property1} as a self-consistent characterization.
% For weighted graphs with independent continuous resistances, this characterization always holds. In unweighted random graphs, configurations in which exact structural symmetries enforce equipotential nodes are non-generic and occur with small probability. Accordingly, the self-consistent approach captures the typical structure of flow subgraphs and can be applied to all random graphs with a known probability generating function (pgf) of node degree $D$.

\subsection{Node number of the flow subgraph in Erd\H{o}s--R\'enyi (ER) random graphs}
\label{sec: Node Number of the flow subgraph in random graphs}
We first analyze how the flow subgraph $G^*_{ij}$ is composed. 
For ER graph $G_p(N)$ with a small expected degree $E[D]<1$, all connected components are small~\cite{newman2010networks,erd6s1960evolution} with size of order $O(\log(N))$.
Hence, in this regime, the relative size of the flow subgraph $\rho_N = |\mathcal N(G^*_{ij})|/N$ is also negligible, where $|\mathcal N(G^*_{ij})|$ denotes the number of nodes in the flow subgraph $G^*_{ij}$.
As the expected degree increases, a giant component~(GC) emerges~\cite{newman2010networks}.
Outside the GC, all connected components have a vanishing relative size and thus the relative flow subgraph size $\rho_N$ remains negligible when either (or both) $i$ and $j$ lie outside the GC.
Therefore, the average size of the flow subgraph $G^*_{ij}$ is dominated by node pairs $(i,j)$ that both belong to the GC.

We then decompose the GC into a backbone subgraph $\mathcal{B}$ and branches. 

\begin{enumerate}
    \item{Backbone $\mathcal{B}$:} The maximal induced subgraph of the (giant) component in which every node has at least two neighbors within $\mathcal{B}$ (nodes highlighted with green outlines in Fig.~\ref{fig:branches_backbone}). 
    Denote by $b=|\mathcal{B}|/N$ the relative size of the backbone $\mathcal{B}$, where $|\mathcal{B}|$ denotes the cardinality of the set $\mathcal{B}$.
    \item Branches: the connected components of \(U:=GC\setminus \mathcal{B}\) (nodes without green outlines in Fig.~\ref{fig:branches_backbone}). 
    Each branch $T_k$ (with index $k = 1,\dots, M$) is a finite subgraph attached to $\mathcal{B}$ at a single node, where $M$ denotes the total number of branches. Let $\theta := |U|/N$ denote the total relative size of the union of all branches.
\end{enumerate}

\begin{figure}[!htbp]
  \centering
  \includegraphics[scale=0.3]{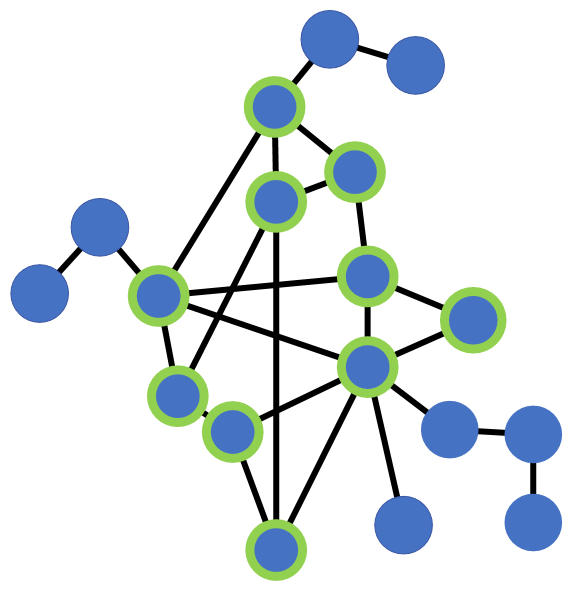}
  \caption{Backbone–branch decomposition: Backbone nodes (each node has $\ge 2$ neighbors within the backbone) are highlighted with green outlines; other nodes are branches (finite subgraph) attached to the backbone.}
\label{fig:branches_backbone}
\end{figure}

\subsubsection{Relative size of the backbone}
\label{sec:Relative size of the backbone}
Consider a flow subgraph $G^*_{ij}$ with source node $i$ and destination $j$.
A node either belongs to the backbone $\mathcal{B}$ or to its complement $\mathcal{B}^c$ with probability $b$ and $1-b$, respectively. 
By definition, any node $m$, different from node $i$ or $j$, can be in $\mathcal{B}$ only if it has at least two neighbors that also lie within the backbone $\mathcal{B}$ by Property~\ref{property1}.

Let $p_b$ denote the probability that a node $m$ reached via a uniformly chosen link has at least one remaining neighbor that belongs to the backbone $\mathcal{B}$.
For each remaining neighbor node $n$ of node $m$, the probability that node $n$ cannot reach node $m$ via link $l = m \sim n$ is $1-p_b$. In the large-$N$ limit, Erd\H{o}s–R\'enyi graphs are locally tree-like (see, e.g., \cite{van2024random}), so correlations between distinct remaining neighbors vanish. Under this approximation, we obtain the self-consistent equation
\begin{align}
    p_b &= 1 - \text{Pr}\left[\text{No remaining neighbor belongs to $\mathcal{B}$}\right] \notag\\
    % & = 1 - \sum_{k}\text{Pr}\left[\text{No remaining neighbor belongs to the backbone $\mathcal{B}$ $|D^\prime=k$}\right]\text{Pr}\left[D^\prime=k\right] \notag\\
    & = 1- \sum_{k}^{N-1}\Pr[D^\prime=k](1-p_b)^k \notag\\
    &= 1 - \varphi_{D^{\prime}}(1-p_b)
    \label{eq:self-consistent-general}
\end{align}
Eq.~\eqref{eq:self-consistent-general} indicates that the probability $p_b$ is identical to the probability that a node belongs to the GC of a random graph~\cite{PhysRevE.64.026118,newman2003structure,PhysRevLett.85.5468}.

Given a node $m$ with degree $D=k$, let $X$ denote the number of neighbors of $m$ that belong to the backbone $\mathcal{B}$. Conditioned on $D=k$, the variable $X$ follows a binomial distribution and
\[
\Pr[X = j \mid D = k] = \binom{k}{j} p_b^j (1-p_b)^{k-j}
\]

A node $w$ belongs to the backbone $\mathcal{B}$ only if at least two of its neighbors belong to $\mathcal{B}$. Hence,
\begin{align}
&\Pr[w \in \mathcal{B} \mid D = k] \notag\\[4pt]
=& \Pr[X \ge 2 \mid D = k] \notag\\[4pt]
=& 1 - \Pr[X = 0 \mid D = k]
     - \Pr[X = 1 \mid D = k]\notag
\end{align}
Using the binomial distribution of $X$,
\begin{equation}
\Pr[w \in \mathcal{B} \mid D = k]
= 1 - (1-p_b)^k - k p_b (1-p_b)^{k-1}
\end{equation}
Applying the law of total probability with respect to the degree, the fraction $b$ of nodes in the backbone $\mathcal{B}$ is
\[
b
= \sum_{k=0}^{N-1}
\Pr[D=k]\Pr[w\in\mathcal{B}\mid D=k]
\]
Substituting $\Pr[w\in\mathcal{B}\mid D=k]
= 1-(1-p_b)^k-kp_b(1-p_b)^{k-1}$ gives
\[
b
= \sum_{k=0}^{N-1}\Pr[D=k]
\Big[1-(1-p_b)^k-kp_b(1-p_b)^{k-1}\Big]
\]
Separating the terms and identifying the resulting sums
with $\varphi_D(1-p_b)$ and $\varphi_D'(1-p_b)$,
we obtain
\begin{equation}
b
= 1-\varphi_D(1-p_b)-p_b\,\varphi_D'(1-p_b)
\label{eq:backbone-size}
\end{equation}

For an ER graph $G_p(N)$ with mean degree
$E[D]=(N-1)p=\lambda$, the degree distribution is binomial~\cite{van2014performance} with
probability generating function
\[
   \varphi_D(z) = (1-p(1-z))^{N-1}
\]
In the sparse regime where $N\to\infty$ and $\lambda$ is fixed, this pgf
admits the expansion~\cite{van2014performance}
\begin{equation}
   \varphi_D(z)
   = e^{-\lambda(1-z)}
     \left(1 + O\!\left(\frac{1}{N}\right)\right)
   \label{eq:phi_ER}
\end{equation}
The degree and excess--degree probability generating functions of an ER graph given by $\varphi_D(z)=(1-p(1-z))^{N-1}$ and $\varphi_{D'}(z)=(1-p(1-z))^{N-2}$, respectively, thus satisfy, $\varphi_{D'}(z) = e^{-\lambda(1-z)}+O(1/N)$.
Substituting this expansion into the self-consistency
equation~\eqref{eq:self-consistent-general} yields
\begin{equation}
   p_b = 1 - e^{-E[D]\,p_b},
   \label{eq:self-consistent-ER}
\end{equation}
where the omitted terms introduce corrections of order $O(1/N)$ for ER graphs.

Eq.~\eqref{eq:self-consistent-ER} can be solved in closed form. Rewrite
\[
p_b = 1-e^{-E[D]p_b}
\quad\Longleftrightarrow\quad
1-p_b = e^{-E[D]p_b}
\]
Let $y=E[D](1-p_b)$. Then
\begin{equation}
    (-y)e^{-y} = -E[D]\,e^{-E[D]}
    \label{eq.lambert_1}
\end{equation}
The Lambert--$W$ function~\cite{corless1996lambert}
is defined implicitly by $W(x)e^{W(x)}=x$.
Solving~\eqref{eq.lambert_1} in terms of the Lambert--$W$ function yields
\begin{equation}
p_b
= 1+\frac{1}{E[D]}\,
W_0\!\left(-E[D]\,e^{-E[D]}\right),
\label{eq.p_b_closed_form}
\end{equation}
where $W_0$ denotes the principal branch.
For $E[D]\ge 0$, this branch gives the unique solution satisfying
$p_b\in[0,1]$, while the remaining real branch leads to non-physical solutions and is discarded.

Using the same sparse regime approximation as in~\ref{eq:phi_ER} in the backbone-size expression
\eqref{eq:backbone-size} by omitting $O\!\left(\frac{1}{N}\right)$ gives
\begin{equation}
   b \;=\; 1 - e^{-E[D]\,p_b}
              \bigl(1+E[D]\,p_b\bigr)
   \label{eq:backbone-size-ER}
\end{equation}

Eq.~\eqref{eq:self-consistent-ER} always has the trivial solution $p_b=0$; for the expected degree $E[D]\le 1$, the trivial solution is the only fixed point, which implies $b=0$ from \eqref{eq:backbone-size-ER}.
When the expected degree $E[D]>1$, a positive fixed point $p_b>0$ appears and increases with $E[D]$, and the backbone fraction $b$ given by \eqref{eq:backbone-size-ER} grows accordingly.

\subsubsection{Relative size of the branches}
\label{subsubsec:branch-size}
For the the union of all branches $U = GC \setminus \mathcal{B}$, the relative size $\theta = \frac{|U|}{N} = p_b - b$, while the relative size of each branch $T_k$ is defined as $\tau_k = \frac{|T_k|}{N}$, where $|T_k|$ is the size of a single branch.

For the ER graph or a configuration model with fixed degree distribution possessing a finite second moment, each branch can be regarded as a subcritical Galton–Watson tree~\cite{newman2010networks,PhysRevE.64.026118,PhysRevLett.85.5468} with offspring mean, denoted by $\xi<1$.
Let $p_T$ denote the extinction probability of the corresponding Galton–Watson process and $p_T$ satisfies $p_T = \varphi_{D^{\prime}}(p_T)$, where $\varphi_{D^{\prime}}(z)$ is the pgf of the excess-degree distribution of the graph.
The average offspring~\cite{newman2010networks,PhysRevE.64.026118,PhysRevLett.85.5468} is then given by $\xi = (1-p_T)\,\varphi_{D^{\prime}}'(p_T)$. 
In the specific case of an ER graph, the average offspring specializes to $\xi = E[D]p_T(1-p_T)<1$.

Since a single branch behaves as a subcritical Galton–Watson tree, the expected size of the branch $T_k$ is constant and independent of the size of the graph $N$: $E[|T_k|] = \frac{1}{1-\xi}$.
The branch size distribution has an exponentially decaying tail~\cite{newman2010networks,PhysRevE.64.026118,PhysRevLett.85.5468}, implying that $|T_k| = O(1)$ and therefore
\[\tau_k = \frac{|T_k|}{N} = O\!\left(\frac{1}{N}\right) \]

Since the total branch mass is $|U|=\theta N$ and a single branch has $E[|T_k|]=\frac{1}{1-\xi}$, the expected number of branches is $E[M] \approx \frac{|U|}{E[\,|T_k|\,]} \;=\; \theta N\, (1-\xi).$
% Thus, individual branches remain of constant size, while their relative size $\tau_k$ decreases to zero as $N$ increases.

\subsubsection{Relative size of the flow subgraph}
\label{sec:Case analysis}
Conditional on source-destination node pair $(i,j)\in{\rm GC}$, three different cases are considered in Fig.~\ref{fig:mu_sizecase}:
\begin{enumerate}
\item{Case (i):} Node $i$ and $j$ lie in different branches or one is in a branch and the other in the backbone $\mathcal{B}$. Case (i) occurs with probability $O(1)$, because the total relative size of all branches is $O(1)$ whereas each individual branch has $O(1)$ size. Hence, two random GC nodes fall in different branches (or one in a branch and one in $\mathcal{B}$) with a non-vanishing probability as $N\to\infty$. The normalized coverage within the backbone converges to $b$, while the branch parts of the flow subgraph $G_{ij}^*$ contain only $O(1)$ nodes and therefore contribute $O(1/N)$ after normalization. This follows from the fact that each branch behaves as a subcritical Galton--Watson tree whose expected size is bounded and independent of $N$ (see discussion above).

\item{Case (ii):} Both node $i$ and $j$ lie within the backbone $\mathcal{B}$. Case (ii) also has probability $O(1)$ under the same conditioning. In this case, the flow remains entirely inside the backbone, and the normalized size of the visited backbone portion converges to $b$.

\item{Case (iii):} Both node $i$ and $j$ belong to the same branch. The flow subgraph $G^*_{ij}$ is contained inside a single finite tree. Since a branch is a subcritical Galton--Watson tree of constant expected size (independent of $N$), we have the size of flow subgraph $|\mathcal N(G^*_{ij})| = O(1)$ and therefore the relative size of the flow subgraph $|\mathcal N(G^*_{ij})|/N = O(1/N)$. The probability of case (iii) is $O(1/N)$, so its contribution to the expected normalized size of the flow subgraph is $O(1/N)\cdot O(1/N)=O(1/N^2)$, which is negligible in expectation.

\end{enumerate}

\begin{figure}[!htbp]
    \centering
    % ---------- 子图 (a) ----------
    \begin{subfigure}[t]{0.16\textwidth}
        \includegraphics[width=\linewidth]{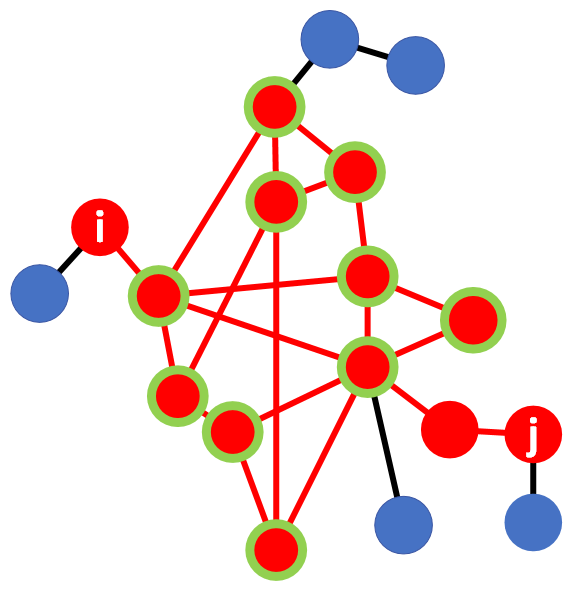}
        \caption{Case (i)}
        \label{sfig:mu_size_case1}
    \end{subfigure}
    \hspace{-5pt}
    % ---------- 子图 (b) ----------
    \begin{subfigure}[t]{0.16\textwidth}
        \includegraphics[width=\linewidth]{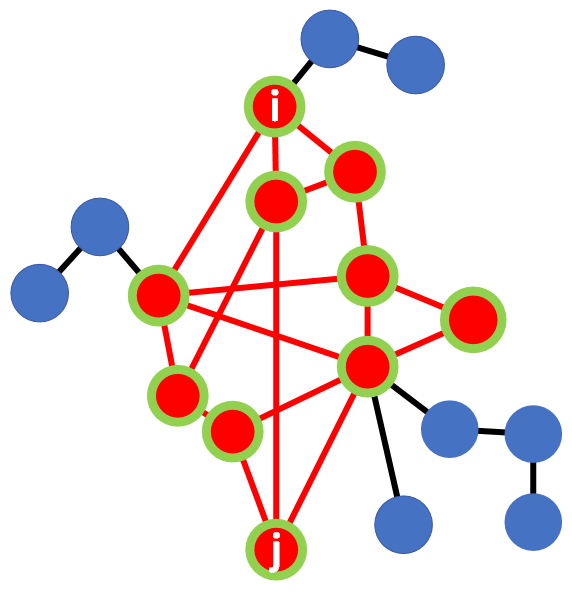}
        \caption{Case (ii)}
        \label{sfig:mu_size_case2}
    \end{subfigure}
    \hspace{-5pt}
    % ---------- 子图 (c) ----------
    \begin{subfigure}[t]{0.16\textwidth}
        \includegraphics[width=\linewidth]{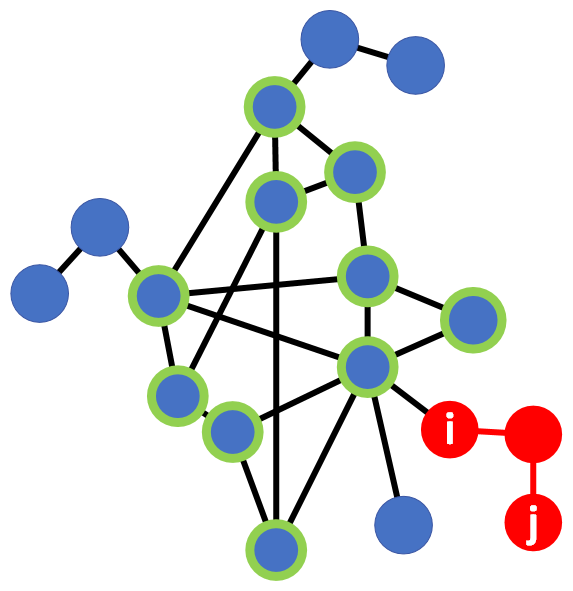}
        \caption{Case (iii)}
        \label{sfig:mu_size_case3}
    \end{subfigure}
    % ---------- 总标题 ----------
    \caption{Examples for different source--destination pairs. Red nodes and links indicate the components of the flow subgraph. Backbone nodes are highlighted with green node outlines.}
    \label{fig:mu_sizecase}
\end{figure}

In Section~\ref{sec:Relative size of the backbone}, we have shown that the probability that a randomly chosen node $i$ belongs to the giant component GC is $\Pr[i\in{\rm GC}]=p_b$. Therefore the probability that both nodes $i$ and $j$
belong to the GC is $\Pr[i,j\in{\rm GC}]=p_b^2$. Conditioning on this event,
Cases (i) and (ii) imply that the normalized size of the flow subgraph
restricted to the GC satisfies
\[
E\!\left[\frac{|\mathcal N(G^*_{ij})|}{N}\,\middle|\, i,j\in{\rm GC}\right]
= b + O\!\left(\frac{1}{N}\right),
\]
where the $O(1/N)$ term comes from the branches on the
GC. Nodes outside the GC occur with probability $1-p_b^2$ and generate a
flow subgraph of size $O(1)$, since every connected component outside the GC is a finite tree of bounded size. Their total contribution to $E[\rho_N]$ is therefore also $O(1/N)$. Consequently, the expected normalized size of the flow subgraph satisfies
\begin{equation}
    \label{eq:expectedsizeofflowsubgraph}
    E[\rho_N]
    = b\,p_b^2 + O\!\left(\frac{1}{N}\right)
\end{equation}

\begin{figure}[!htbp]
  \centering
  \includegraphics[width=0.7\linewidth]{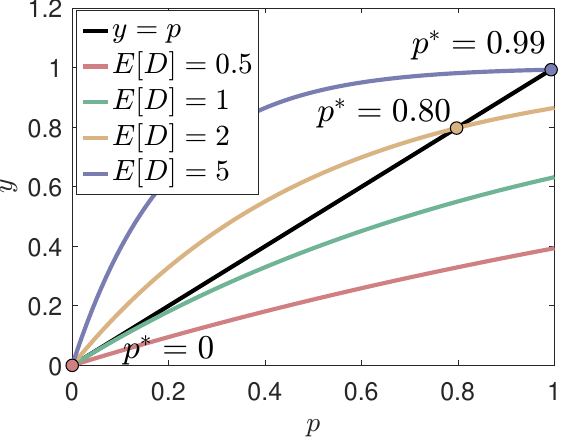}
  \caption{Intersections of $y=p$ (black) and $y=1-e^{-E[D] p}$ (colored) for different expected degrees $E[D]$. When the expected degree $E[D] \le 1$, only the trivial solution $p=0$ exists. For $E[D]>1$, a non-trivial fixed point appears and increases with $E[D]$.}
\label{Fig.ER_intersections}
\end{figure}

\begin{figure}[!htbp]
    \centering
    % ---------- 左子图 ----------
    \begin{subfigure}[t]{0.8\linewidth}
        \includegraphics[width=\linewidth]{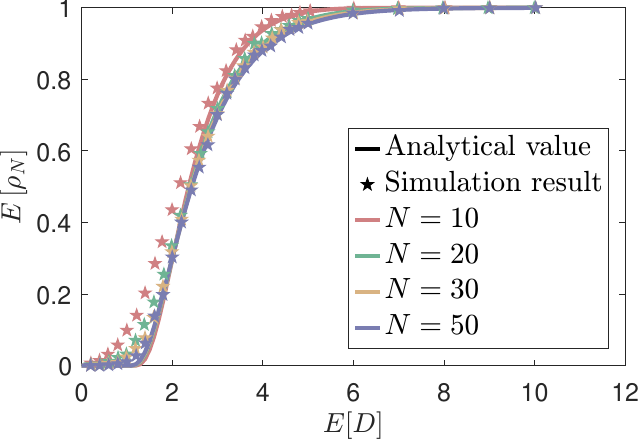}
        \caption{$N$ from 10 to 50}
        \label{sfig:mu_size_small}
    \end{subfigure}
    % ---------- 右子图 ----------
    \begin{subfigure}[t]{0.8\linewidth}
        \includegraphics[width=\linewidth]{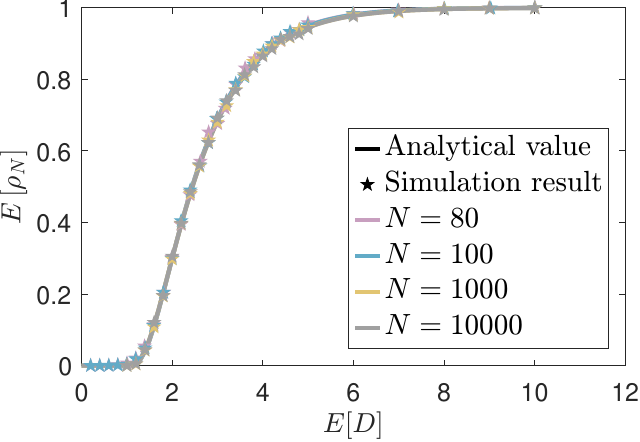}
        \caption{$N$ from 80 to 10000}
        \label{sfig:mu_size_large}
    \end{subfigure}

    % ---------- 总标题 ----------
    \caption{The average relative size $E\left[\rho_N\right]$ of the flow subgraph as a function of the expected degree $E[D]$ in ER graphs with identical link weights for different graph sizes $N$. 
    The analytical results in \eqref{eq:expectedsizeofflowsubgraph} are derived using the binomial degree distribution of the ER graph.}
    \label{fig:mu_size}
\end{figure}

Fig.~\ref{Fig.ER_intersections} shows that when the expected degree $E[D]$ is small (e.g., $E[D]=0.5$), the two curves $y=p$ and $y=1-e^{-E[D] p}$ intersect only at $p=0$, Eq.~\eqref{eq:self-consistent-ER} has only the solution 0, so the normalized size $b$ of the backbone $\mathcal{B}$ is zero. When the expected degree $E[D]$ exceeds the critical value $(E[D])_c = 1$, a second intersection emerges, indicating that a giant flow subgraph starts to form. As the expected degree $E[D]$ increases further (e.g., $E[D]=2, 5$), the fixed point $p$ and the corresponding size $b$ both grow and eventually approach 1.

Fig.~\ref{fig:mu_size} compares our analytical prediction in
\eqref{eq:expectedsizeofflowsubgraph} with simulations on ER graphs. For small network sizes $N$, a noticeable gap appears between the theoretical curve and the simulation results, which is consistent with the $O(1/N)$
finite-size deviation obtained in our analysis. As $N$ increases, this finite-size effect diminishes, and the analytical prediction in
\eqref{eq:expectedsizeofflowsubgraph} converges to the simulation results. When the graph is sufficiently large, the $O(1/N)$ correction becomes negligible, and the agreement between theory and simulation is nearly exact.

\subsection{Link number of the flow subgraph in Erd\H{o}s--R\'enyi (ER) random graphs}\label{sec:Link numbers of flow subgraphs}
For a source-destination node pair $(i,j)$, a link $l=m\sim n$ belongs to the flow subgraph $G^*_{ij}$ if and only if it carries a nonzero current, which by Ohm’s law is equivalent to $v_m\neq v_n$. Every link in the flow subgraph $G^*_{ij}$ must have both endpoints in the node set $\mathcal{N}(G^*_{ij})$ of the flow subgraph. Among the links belonging to $G^*_{ij}$, the membership condition is equivalent to
\[
l = m\sim n \in \mathcal{L}(G^*_{ij})
\quad\Longleftrightarrow\quad
y_{l} \neq 0
\quad\Longleftrightarrow\quad
v_m\neq v_n
\]
Identifying the links in the flow subgraph reduces to considering links whose endpoints both belong to the node set $\mathcal{N}(G^*_{ij})$ of the flow subgraph and have different electrical potentials $v$.

% Hence, any node $i \notin \mathcal{N}(G^*_{ij})$ carries no current, and Kirchhoff’s current law implies that all its incident link currents must be zero: $y_{ij}=0$ for every neighbor $v\sim u$. 
% Consequently, no link adjacent to a node outside the node set of the flow subgraph $N(G^*_{ij})$ can belong to the flow subgraph. In particular, every link
% in the flow subgraph $G^*_{ij}$ must have both endpoints in the node set of flow subgraph $N(G^*_{ij})$, and among those links the membership condition reduces to
% \[
% u\sim v \in L(G^*_{ij})
% \quad\Longleftrightarrow\quad
% I_{uv}\neq 0
% \quad\Longleftrightarrow\quad
% v_u\neq v_v.
% \]
% Therefore, when counting the links in the flow subgraph, we only need to
% consider links whose endpoints both lie in the node set of flow subgraph $N(G^*_{ij})$ and have
% different electrical potentials.

\begin{figure}[!htbp]
    \centering
    % ---------- Subfigure 1 ----------
    \begin{subfigure}[t]{0.4\linewidth}
        \captionsetup{width=\linewidth}  % 调整 caption 宽度（比图片略宽）
        \includegraphics[width=0.8\linewidth]{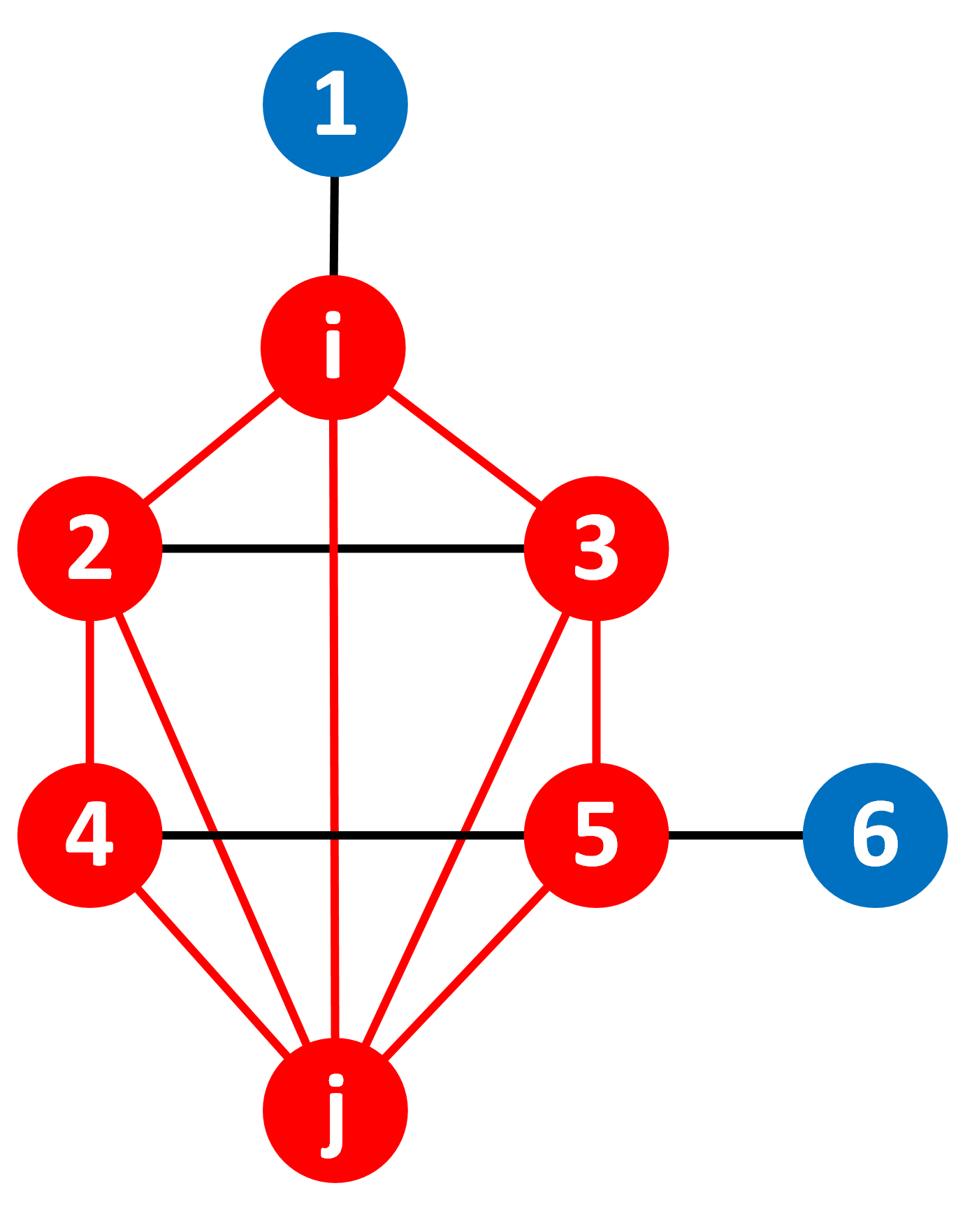}
        \caption{Flow subgraph $G_{ij}^*$ with $N=6$ nodes and $L=9$ links.}
        \label{sfig:flowsubgraphtoymodel1}
    \end{subfigure}
    % ---------- Subfigure 2 ----------
    \begin{subfigure}[t]{0.4\linewidth}
        \captionsetup{width=\linewidth}  % 调整 caption 宽度（标题更长）
        \includegraphics[width=0.6\linewidth]{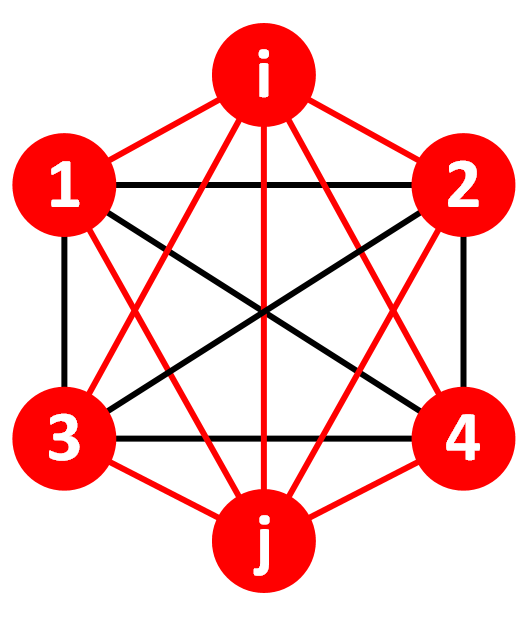}
        \caption{Flow subgraph $G_{ij}^*$ of a complete graph with $N=6$ nodes.}
     \label{sfig:flowsubgraphtoymodel2}
    \end{subfigure}
    \caption{Schematic illustration of links in flow subgraphs. The flow subgraph is highlighted in red. In (a), nodes 2, 3, 4, and 5 belong to the flow subgraph. Links $1 \sim i$ and $5 \sim 6$ are not included because nodes 1 and 6 do not belong to the flow subgraph. Links $2 \sim 3$ and $4 \sim 5$ are also excluded, since the node pairs $(2,3)$ and $(4,5)$ have equal potentials. In (b), all links between nodes 1, 2, 3, and 4 are not flow subgraph links.
    }
    \label{fig:unweighted}
\end{figure}

To investigate the link set $\mathcal{L}(G^*_{ij})$ of the flow subgraph, we first define a set of links $\mathcal{L}_{\mathrm{FS}}(i,j)$, referred to as \textit{candidate links},
% We first define the set of \textit{candidate links}, which connects nodes in the flow subgraph $G^*_{ij}$,
\[
\mathcal{L}_{\mathrm{FS}}(i,j)
    := \{\, l = m\sim n \in \mathcal{L}(G)\;|\; m,n \in \mathcal{N}(G^*_{ij}) \,\},
\]
The set $\mathcal{L}_{\mathrm{FS}}(i,j)$ contains all links whose endpoints belong to the flow-subgraph node set $\mathcal{N}(G^*_{ij})$.
However, in graphs with identical link weights, a link $l = m\sim n\in\mathcal{L}_{\mathrm{FS}}(i,j)$ may connect two equipotential nodes $m$ and $n$ and therefore carry zero current, so that link $l$ does not belong to the flow subgraph (see examples in Figs.~\ref{sfig:unweighted_flow_ex} and~\ref{sfig:flowsubgraphtoymodel1}).

% Since a node pair $(m,n)$ belonging to the flow subgraph may have the same potential, a link $l=m\sim n$ may belong to $\mathcal{L}_{\mathrm{FS}}(i,j)$, but not to  $\mathcal{L}(G^*_{ij})$, Fig.~\ref{sfig:flowsubgraphtoymodel1}. 
By construction, every flow-subgraph link must connect two flow-subgraph nodes, whereas the converse need not hold. Thus, the link set of the flow subgraph satisfies
\begin{equation}
\label{eq:samepotential}
    \mathcal{L}(G^*_{ij}) \subseteq \mathcal{L}_{\mathrm{FS}}(i,j)
\end{equation}
% Among the links in the candidate link set $\mathcal{L}(G^*_{ij})$ of flow subgraph, the only ones that do not carry
% current are those whose endpoints become exactly equipotential under the
% $i\!\to\! j$ injection. 
For weighted ER graphs with link weights drawn independently from a continuous distribution, nodal potentials are distinct (Lemma~\ref{lem:weighted_prob0_equipotential}), which implies $\mathcal{L}(G^*_{ij}) = \mathcal{L}_{\mathrm{FS}}(i,j)$.
In ER graphs with identical link weights, however, equipotential nodes may arise from local structural symmetries, namely, configurations in which two distinct nodes $m$ and $n$ have identical adjacency patterns with respect to the rest of the graph. In other words, there exists a graph automorphism~\cite{erdos1963asymmetric} that can map node $m$ onto $n$. In such symmetric configurations, the two nodes attain the same potential, so that a candidate link between them carries zero current and does not belong to $\mathcal{L}(G^*_{ij})$.

For ER graphs with identical link weights and moderate expected degree, such nontrivial automorphisms occur~\cite{erdos1963asymmetric} with probability~$O(1/N)$, giving
\[
|\mathcal{L}_{\mathrm{FS}}(i,j)|-|\mathcal{L}(G^*_{ij})|=O(1),\
\frac{|\mathcal{L}(G^*_{ij})|}{|\mathcal{L}_{\mathrm{FS}}(i,j)|}=1-O(1/N)
\]
However, when the expected degree $E[D]$ becomes very large (approaching a complete graph), many nodes have almost the same neighbors, making them structurally indistinguishable and thus equipotential. 
The number of such equipotential node pairs increases rapidly with $p$ and a fraction of links in $\mathcal{L}_{\mathrm{FS}}(i,j)$ carry zero current. 
Consequently, the difference $|\mathcal{L}_{\mathrm{FS}}(i,j)| - |\mathcal{L}(G^*_{ij})|$ is no longer $O(1)$ but grows proportionally to the number of symmetric (and equipotential) node pairs.

\begin{figure}[!htbp]
  \centering
  \includegraphics[scale=0.67]{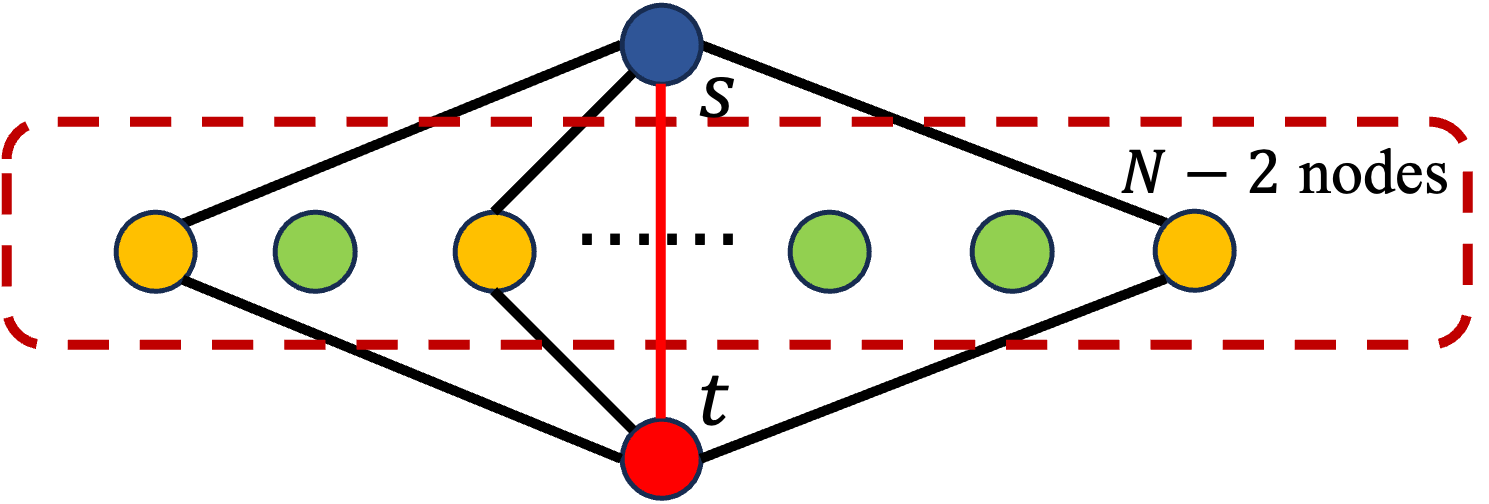}
  \caption{
Illustration of two nodes $s$ and $t$ with identical neighbor sets.
Out of the remaining $N-2$ nodes in the network, $k$ nodes (yellow) are connected to both $s$ and $t$, while the remaining $N-2-k$ nodes (green) are not connected to either $s$ or $t$. The links between nodes other than $s$ and $t$ are invisible for clarity.
}
\label{fig:same_neighbor_ex}
\end{figure}

We now derive an upper bound for the expected link number of the flow subgraph in ER graphs $G_p(N)$ with identical link weights. Consider an unordered node pair $(s,t)$ and condition on the existence of the link $l=s\sim t$. As illustrated in Fig.~\ref{fig:same_neighbor_ex}, nodes $s$ and $t$ share an identical adjacency pattern with respect to the remaining $N-2$ nodes. Each node other than nodes $s$ and $t$ is either connected to both $s$ and $t$ (a common neighbor of $s$ and $t$) or to neither of them. If nodes $s$ and $t$ are not current source or sink nodes, the symmetry then leads to equal potentials at $s$ and $t$ and the link $l=s\sim t$ carries zero current.

Suppose exactly $k$ out of the $N-2$ nodes are common neighbors of $s$ and $t$ in the configuration shown in Fig.~\ref{fig:same_neighbor_ex}, while the remaining $N-2-k$ nodes are connected to neither $s$ nor $t$.
Ignoring the connectivity among nodes other than $s$ and $t$, the conditional probability of this configuration (given the existence of $l = s \sim t$) is $\binom{N-2}{k} (p^2)^k \big((1-p)^2\big)^{N-2-k}$. 
Summing over $k=0,\dots,N-2$, the probability that a present link joins two nodes with identical neighbor sets is
\begin{equation}
\begin{aligned}
\Pr[\mathcal{E}_{st}]
&=  \sum_{k=0}^{N-2}
   \binom{N-2}{k}
   (p^2)^k
   \big((1-p)^2\big)^{N-2-k} \\
&= \big(p^2+(1-p)^2\big)^{N-2}
\end{aligned}
\end{equation}

Since there are $\binom{N}{2}$ unordered node pairs in an ER graph $G_p(N)$ with identical link weights, the expected fraction of such present links equals $\Pr[\mathcal{E}_{st}]$. 
The identical neighbor sets (Fig.~\ref{fig:same_neighbor_ex}) constitute only one type of configuration that leads to zero current on $l = s \sim t$. 
Other graph automorphisms may also render a link without a current.
Therefore, for large $N$, the expected fraction $E[\rho_L]$ of links in the flow subgraph, normalized by the total number of links in an ER graph with identical link weights, is upper bounded by
\begin{equation}
1 - \big(p^2+(1-p)^2\big)^{N-2}
\end{equation}

In a complete graph $K_N$ with identical link weights, all nodes other than source-destination nodes $i$ and $j$ have exactly the same neighbors and hence are equipotential. 
Hence, every link whose endpoints do not include $i$ or $j$ has zero potential difference and carries no current. 
The only links with nonzero potential differences are:
(i) the link $l=i \sim j$ and (ii) the $2(N-2)$ links on the two-hop paths $\mathcal{P}_{ij} = \{i\sim k,k \sim j \}$, $k$ is different from node $i$ and $j$. Thus, almost all links of $K_N$ lie outside the flow subgraph.

\begin{figure}[!htbp]
    \centering

    % ---- 左子图 ----
    \begin{subfigure}[t]{0.8\linewidth}
        \includegraphics[width=\linewidth]{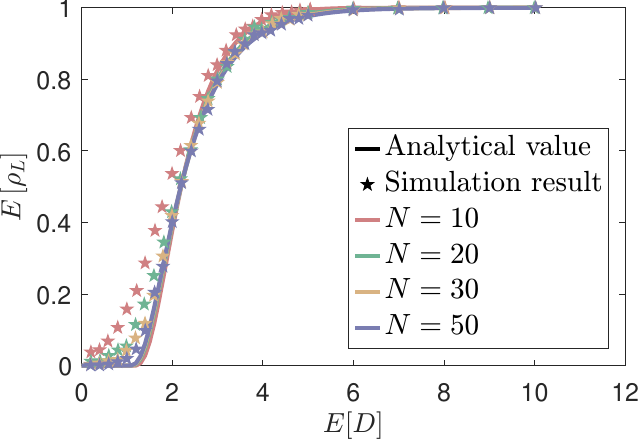}
        \caption{$N$ from 10 to 50}
        \label{sfig:mu_link_small}
    \end{subfigure}
    \hfill
    % ---- 右子图 ----
    \begin{subfigure}[t]{0.8\linewidth}
        \includegraphics[width=\linewidth]{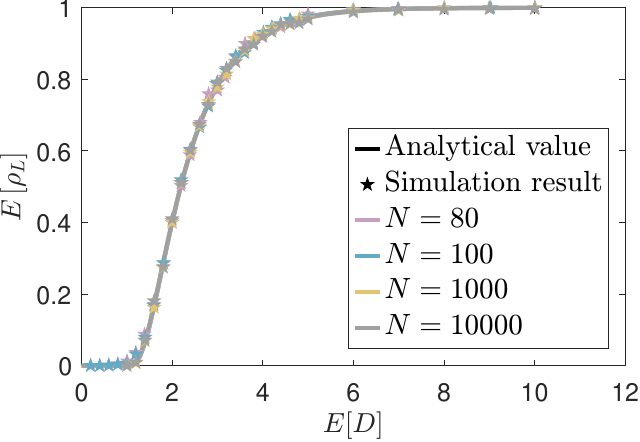}
        \caption{$N$ from 80 to 10000}
        \label{sfig:mu_link_large}
    \end{subfigure}

    % ---- 总标题 ----
    \caption{The average link fraction $E\left[\rho_L\right]$ of the flow subgraph as a function of the expected degree $E[D]$ in ER graphs with i.i.d. continuous link weights for different numbers of nodes $N$.}
    \label{fig:mu_link}
\end{figure}

Finally, we present the analytic formula of the expected fraction of links $E[\rho_L]$ in the flow subgraph. In Section~\ref{sec:Case analysis}, we have shown that as the size of the graph \(N\) increases, for a fraction $1-p_b^{2}$ of source-destination node pairs \((i,j)\), which lie outside GC, the relative size $\rho_N$ of the flow subgraph satisfies
\[
\rho_N = \frac{|\mathcal N(G^*_{ij})|}{N} \longrightarrow 0,
\]
which further indicates that the expected fraction $\frac{|\mathcal{L}_{\mathrm{FS}}(i,j)|}{L}$ of candidate links $\mathcal{L}_{\mathrm{FS}}(i,j)$, normalized by the total number of links $L$ of the ER graph tends to 0 with $N\to\infty$.
For the remaining fraction \(p_b^{2}\) of node pairs, the relative size of the flow subgraph approaches the size of the backbone \(\mathcal{B}\), i.e.,
\[
\rho_N \longrightarrow \frac{|\mathcal{B}|}{N}
\]
For a fixed source–destination pair $(i,j)$ that lies in the backbone $\mathcal{B}$, each endpoint of a uniformly chosen link belongs to the backbone with probability $p_b$. Hence, conditional on $(i,j)\in\mathcal B$, a link belongs to the candidate link set $\mathcal L_{\mathrm{FS}}(i,j)$ with probability $p_b^2$. Equivalently,
\[
E\!\left[\frac{|\mathcal L_{\mathrm{FS}}(i,j)|}{L}\;\middle|\; (i,j)\in\mathcal B\right]
= p_b^2 
\]

Averaging over all source–destination pairs $(i,j)$, we note that $\Pr[(i,j)\in\mathcal B]=p_b^2$, while for pairs outside the backbone, the flow subgraph is empty. Therefore, the overall expected fraction of links in the flow subgraph satisfies
\begin{equation}
\label{eq: expected link size of flowsubgraph}
\begin{aligned}
E[\rho_L]
&= p_b^2 \times p_b^2 + (1-p_b^2)\times 0 + O\!\left(\frac{1}{N}\right) \\
&= p_b^4 + O\!\left(\frac{1}{N}\right)
\end{aligned}
\end{equation}

\begin{figure}[!htbp]
  \centering
  \includegraphics[width=0.95\linewidth]{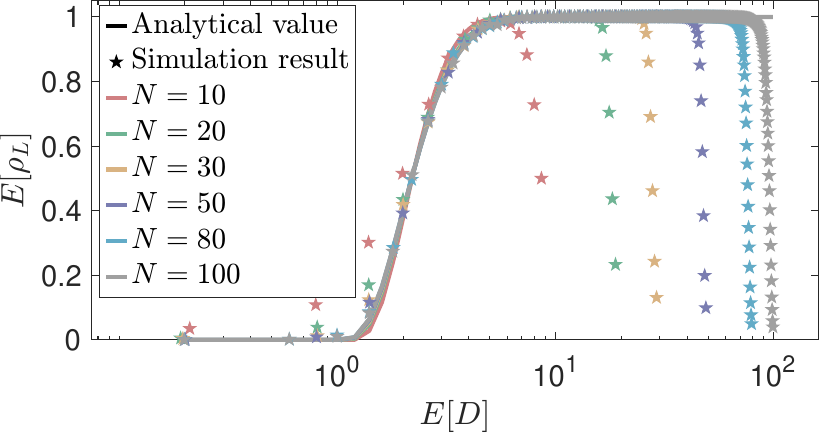}
  \caption{The average link fraction $E\left[\rho_L\right]$ of flow subgraph as a function of the expected degree $E[D]$ in ER graphs with identical link weights for different numbers of nodes $N$.}
\label{sfig:linksize_flow_subgraph_unweighted}
\end{figure}

Fig.~\ref{fig:mu_link} compares our analytical results in \eqref{eq: expected link size of flowsubgraph} with simulations on weighted ER graphs. 
The link weights are uniformly distributed from $(0,1)$.
A similar pattern to Fig.~\ref{fig:mu_size} is visible.
When the graph size $N$ is small, the branches lead to a gap between the analytical results and the simulation. 
The gap decreases and approaches 0 as $N$ increases.
For ER graphs with identical link weights, our \eqref{eq: expected link size of flowsubgraph} provides accurate predictions as long as the expected degree is not too large.
In the regime of large expected degrees, however, the number of links in the flow subgraph decreases due to \eqref{eq:samepotential}.
In the extreme case where the equal-link-weight graph is complete, the size of the flow subgraph reaches $|G_{ij}^*| = 1+2(N-2)$.

\section{Graph construction with end-to-end demands on power dissipation}
\label{sec:Network construction with end-to-end demands on power dissipation}
In an electrical resistor network $G$, the current flowing through each link leads to power dissipation because of the resistors, \eqref{eq:graphpowerdissipation} and \eqref{eq:graphpowerdissipation2}.
However, the inverse problem remains a non-trivial challenge: how to construct a graph that satisfies predetermined graph power dissipation demands $d_{ij}$, i.e., the required transportation cost between a node pair $(i,j)$ in a flow network.
For example, low-power Internet of Things (IoT) networks~\cite{1717439} often impose demands on the energy dissipation required for transmitting one unit of data between two nodes.
In this section, we focus on the flow network design problem, where we draw an analogy to electrical resistor networks.

Since the demand $d_{ij}$ can be normalized to the graph power dissipation $P_G$ with unit current injected in the graphs, we consider $I_c = 1$ in the following analysis. 
In a path network where the item travels along the shortest path and the link weight $w_l$ denotes the \quotes{resistance} or transportation cost on link $l$ during the transportation, the power dissipation $P_G$ then numerically equals the shortest path weight $s_{ij}$ between node $i$ and $j$.
Hence, the problem reduces to constructing a graph such that the corresponding shortest path weight matrix $S$ equals the given demands $D$, which is solved by the inverse all shortest path problem~\cite{qiu2023inverse}.

Similar to the path network, when unit currents are considered, the inverse power dissipation problem on a flow network reduces to constructing a network that satisfies prescribed end-to-end effective resistance demands:
\begin{problem}[Inverse Effective Resistance   Problem (IERP)]\label{Problem:IERP}
Given an $N\times N$ symmetric demand matrix $D$ with zero diagonal elements but positive off-diagonal elements. Determine an $N\times N$ weighted adjacency matrix $\widetilde{A}$ for the flow network, such that the norm $\|D-\Omega\|$ between the demand matrix and the effective resistance matrix is minimised.
\end{problem}

\subsection{Inverse effective resistance problem}
\label{sec:Inverse effective resistance problem}
Since the effective resistance matrix is a distance matrix\footnote{Any element $h_{ij}$ of a distance matrix $H$ is non-negative $h_{ij}\geq0$, but $h_{ii}=0$ and $h_{ij}$ obeys the triangle inequality: $h_{ij}\leq h_{ik}+h_{kj}$.}, we assume that the demand matrix $D$ is also a distance matrix in the sequel.
Indeed, Van Mieghem \cite{Mieghem2021atree} has proposed a method to transform an arbitrary demand matrix $D$ into a distance matrix:
If $d_{ik} +d_{kj} < d_{ij}$ for at least one node $k \in \mathcal{N}$ which violates the triangle inequality of a distance matrix, then we can replace $d_{ij} = \min_{1\leq k\leq N} (d_{ik} +d_{kj})$ and $d_{ji} = d_{ij}$.

When the given demand matrix $D$ can be realized by a graph, i.e., there exists at least one graph whose effective resistance matrix $\Omega = D$, the IERP can be solved elegantly by Fiedler's block matrix relation \cite{fiedler1998some,fiedler2011matrices,Mieghem2021atree},
\begin{equation}
	{\left( \begin{array}{cc}
			0 & u^T \\
			u & \Omega \\
		\end{array} 
		\right )}^{-1}={\left( \begin{array}{cc}
			-2\sigma^2 & p^T \\
			p & -\frac{1}{2}\widetilde{Q} \\
		\end{array} 
		\right )}
	\label{Fiedler matrix}
\end{equation}
with $\Omega p=2\sigma^2u$, where $\widetilde{Q}=\widetilde{\Delta}-\widetilde{A}$ is the weighted Laplacian matrix of a flow network and the diagonal matrix $\widetilde{\Delta} = {\rm diag}(\widetilde{A}u)$, $\widetilde{A}$ is the weighted adjacency matrix of a flow network, the variance $\sigma^2=\frac{\zeta^T\widetilde{Q}\zeta}{4}+R_G$, where $R_G$ is the effective graph resistance~\cite{PVM_GraphSpectra2023} and $u$ is $N\times1$ the all-one vector.
The vector $\zeta$ contains the diagonal elements of pseudoinverse $\widetilde{Q}^{\dag}$ of the Laplacian $\widetilde{Q}$. 
% Specifically, Van Mieghem\cite{Mieghem2021atree} defined the weight of a link $w_{ij}=r_{ij}$ as the resistance (in Ohm), then the weighted Laplacian $\widetilde{Q}$ has non-zero elements $\widetilde{q}_{ij}=-\frac{1}{r_{ij}}$ and $\widetilde{a}_F = \frac{1}{r_{ij}}$ for $i\neq j$, where $\widetilde{a}_{ij} = r_{ij}$, but $(\widetilde{a}_F)_{ij}=(\widetilde{a})_{ij}=0$ for $i\neq j$ if there is no link between node $i$ and node $j$. The diagonal elements $(\widetilde{a}_F)_{ii}=(\widetilde{a})_{ii}=0$ are always zero. 
Fiedler’s block matrix relation (\ref{Fiedler matrix}) indicates a one-to-one relation between the effective resistance matrix $\Omega$ and the weighted Laplacian $\widetilde{Q}$ and therefore, also between the effective resistance matrix $\Omega$ and the weighted adjacency matrices $\widetilde{A}$. 
By applying block inverse formulae \cite{PVM_GraphSpectra2023} to Fiedler's block matrix relation, the Laplacian is
\begin{equation}
	\widetilde{Q}=\frac{2}{2\sigma^2}pp^T-2\Omega^{-1}
	\label{equation:getQtilde2}
\end{equation}
where $2\sigma^2=\frac{1}{u^T\Omega^{-1}u}$ and the vector $p=\frac{1}{u^T\Omega^{-1}u}\Omega^{-1}u$.
With $\widetilde{Q}=\widetilde{\Delta}-\widetilde{A}$, the weighted adjacency matrix follows as
\begin{equation}
\widetilde{A}=\widetilde{\Delta}+2\Omega^{-1}-\frac{1}{\sigma^2}pp^T
	\label{equation:getAtilde}
\end{equation}
We refer to the method that substitutes the effective resistance matrix $\Omega$ in \eqref{equation:getAtilde} with a given demand $D$ as the \quotes{Fiedler approach}.
Given a demand matrix $D$, if there exists at least one graph whose effective resistance matrix $\Omega = D$, then the Fiedler approach can provide an exact solution to the IERP instance.
% \ref{equation:getAtilde} then provides an exact solution to the IERP for any given graph-realizable demand matrices $D$.
% We refer to the method which executes \eqref{equation:getAtilde} by substituting the effective resistance matrix $\Omega$ with an input demand matrix $D$, as a reference.

\subsection{Graph realizability of perturbed effective resistances}
We call a demand matrix $D$ graph-realizable if there exists at least one graph $G$ whose effective resistance matrix $\Omega = D$.
Constructing a demand matrix that is guaranteed to be graph realizable remains, to the best of our knowledge, a challenging problem.
The main difficulty lies in the dependence among the effective resistances of different node pairs.
Qiu et al.\cite{qiu2023inverse} illustrated that the graph realizability of effective resistances is highly sensitive to perturbations.
As shown in \cite{qiu2023inverse}, even an insignificant perturbation to a realizable effective resistance matrix $\Omega$ can result in a perturbed matrix $\Omega^\prime$ that is no longer realizable.
Specifically, the weighted adjacency matrix $\widetilde{A}^\prime$, derived from a perturbed effective resistance $\Omega^\prime$ via \eqref{equation:getAtilde}, can have negative off-diagonal elements due to the complicated dependencies between effective resistances.

A similar situation, where small perturbations in the effective resistance matrix can render the matrix not graph-realizable, appears in the network version of Calderón’s inverse problem~\cite{Alessandrini01011988,carmona2024stable,GERNANDT202229}, which is known to be a severely ill-posed problem.
Consider an electrical network $G$ whose link weight represent a conductance and nodes are classified as boundary nodes and interior nodes.
Given the adjacency matrix $A$ of network $G$, the prescribed electric potential on the boundary nodes and the currents injected into (or extracted from) the boundary nodes, Calderón’s inverse problem asks for obtaining the conductance of every link such that all potential and current constraints are satisfied.
Carmona et al.~\cite{carmona2024stable} demonstrated that a small perturbation on the given potentials or currents may result in extremely large variations in the recovered conductances and in some cases may even produce negative conductances.
Moreover, the big discrepancies appear on links that are far away from the boundary nodes in terms of hopcount.
The severe ill-posedness arises because the boundary measurements on currents and voltages are only weakly sensitive to variations of conductances located deep inside the network, so that small perturbations at the boundary input may correspond to large and non-unique changes in the interior conductances.

% Without knowledge of the underlying graph, directly constructing an effective resistance matrix $\Omega$ (as a demand matrix $D$) whose corresponding weighted adjacency matrix is guaranteed to be non-negative remains an open problem.
% The investigation of small perturbations on effective resistance highlights the difficulty of constructing a graph-realizable demand matrix: even an insignificant perturbation of a valid effective resistance matrix may destroy its graph realizability; therefore, generating demands by perturbing an existing effective resistance matrix or by directly specifying a resistance matrix is unreliable.
The investigation of small perturbations of effective resistance indicates that even insignificant changes can destroy the graph realizability of an effective resistance, making demand construction based on perturbing or directly specifying effective resistance matrices unreliable.
Hence, the Fiedler approach, which directly executes \eqref{equation:getAtilde} by substituting the effective resistance matrix $\Omega$ with the given demand matrix $D$, loses feasibility in practical scenarios where measure accuracy is limited.
To solve the IERP for arbitrary demands $D$, we propose an algorithm called \quotes{Resistor Gap Pruning (RGP)} (Algorithm~\ref{alg:RGP}), by leveraging resistor and parallel circuit rules, in Section~\ref{sec:RGP}.

\subsection{Resistor Gap Pruning Algorithm}\label{sec:RGP}
The Resistor Gap Pruning (RGP) algorithm aims to obtain a graph whose effective resistance matrix $\Omega$ is approximately equal to an arbitrary predetermined demand matrix $D$.
Consider a flow network $G$ with each link $l = i \sim j$ possessing a resistor $r_{ij} = \frac{1}{w_{ij}}$ equaling the reciprocal of the link weight. 
In a path network, Qiu et al. \cite{qiu2023inverse} demonstrated that a link with small link weight $w_{ij}$, which corresponds to a high resistor $r_{ij}$ in a flow network, is unlikely to be included in a shortest path.
Utilizing this principle, Qiu et al. \cite{qiu2023inverse} proposed the Omega Link Removal algorithm to construct a graph such that the corresponding shortest path weight matrix $S$ approximately equals a given demand matrix $D$.
Inspired by \cite{qiu2023inverse}, we now propose our RGP algorithm for designing a flow network. 

\begin{figure}[tbp]
\centering
\begin{minipage}{\linewidth}
\caption{Resistor Gap Pruning (RGP)}
\label{alg:RGP}
\begin{algorithmic}[1]
    \Require $N\times N$ demand matrix $D$
    \Ensure $N\times N$ weighted adjacency matrix $\widetilde{A}$ for flow network
    \State $A \gets J - I$
    \State $W \gets$ element-wise reciprocal of $D$ over nonzero elements
    \State $\widetilde{A} \gets A \circ W$
    \State $\Omega \gets$ effective resistance matrix of $\widetilde{A}$
    \State $\epsilon \gets \lVert D-\Omega\rVert$
    \Repeat
        \State $\epsilon^\prime \gets \epsilon$
        \State $\hat{\Omega} \gets$ element-wise reciprocal of $\Omega$ over nonzero elements
        \State $\Gamma \gets (\hat{\Omega}-\widetilde{A}) \circ (D-\Omega)\circ A$
        \State $(i,j) \gets$ indices of the maximum element in $\Gamma$
        \State $A \gets A - e_i e_j^{T} - e_j e_i^{T}$
        \State $\widetilde{A} \gets A \circ W$
        \State $\Omega \gets$ effective resistance matrix of $\widetilde{A}$
        \State $\epsilon \gets \lVert D-\Omega\rVert$
    \Until{$\epsilon<\epsilon^\prime$}
    \State $A \gets A + e_i e_j^{T} + e_j e_i^{T}$
    \State $\gamma_{ij} \gets d_{ij}/\omega_{ij}$ for each non-diagonal $(i,j)$ in $\Omega$
    \State $\alpha \gets$ mean of all $\gamma_{ij}$
    \State $\widetilde{A} \gets \frac{1}{\alpha} A \circ W$
    \State \Return $\widetilde{A}$
\end{algorithmic}
\end{minipage}
\end{figure}

The main mechanism of RGP is based on the parallel circuit rules and \eqref{Eq:Omega_ij_parallel}.
The parallel circuit rules show that the effective resistance in the circuit is always less than or equal to the smallest resistance (see \cite{qiu2023inverse} and Section~\ref{sec:Omega}).
In other words, for two adjacent nodes $i$ and $j$, the effective resistance always satisfies $\omega_{ij} \leq r_{ij}$, with equality holding only when the direct link $i \sim j$ is the unique path connecting nodes $i$ and $j$.
Hence, for an arbitrary node pair $(i,j)$, adding links in $G$ can only decrease the effective resistance $\omega_{ij}$ or leave $\omega_{ij}$ unchanged, while removing links in $G$ \emph{cannot} decrease the effective resistance $\omega_{ij}$.
Furthermore, \eqref{Eq:Omega_ij_parallel} indicates that removing a link with a high resistor\footnote{A high resistor means small current through the link by the law of Ohm; hence, a small proportion of traffic. In the extreme case where the resistor $r_{ij}\rightarrow \infty$ implying that almost no current flows through the link $l = i \sim j$.} $r_{ij}$ has an insignificant effect on the effective resistance $\omega_{ij}$.
The link $l = i \sim j$ with a high resistor $r_{ij}$ is then regarded as a \quotes{redundant effective resistance link}.
The main idea of RGP is to remove redundant effective resistance links from a complete graph so that the effective resistance matrix $\Omega$ approaches the target demand $D$:
\begin{enumerate}
    \item Construct a complete graph whose weighted adjacency matrix $\widetilde{A} = D$.
    \item Iteratively remove links whose removal has an insignificant impact on the graph effective resistance.
\end{enumerate}
The details are shown by Algorithm~\ref{alg:RGP}.

The RGP algorithm is initialised in line 1 by the complete graph with the adjacency matrix $A = J - I$, while the link weight matrix $W$ is element-wise reciprocal of $D$ over nonzero elements, line 2.
Lines 3 and 4 then obtain the weighted adjacency matrix $\widetilde{A}$ of the flow network and the corresponding effective resistance matrix $\Omega$.
The difference between the demand matrix $D$ and the effective resistance matrix $\Omega$ is computed in line $5$ as $\epsilon$.
Specifically, the difference is measured by $L1$ norm $\| D-\alpha\Omega\| = \sum_{ij} |d_{ij}-\alpha\omega_{ij}|$.
Assume there exists a graph $G$ whose effective resistance matrix $\Omega_G$ exactly equals the demand matrix $D$.
The effective resistance matrix $\Omega<\Omega_G$ because the resistance of each link $r_{ij} = d_{ij}$ and there are multiple paths between node $i$ and $j$.
% The remaining RGP aims to iteratively remove redundant effective resistance links and to increase the effective resistance matrix $\Omega$ to better approximate the demand matrix $D$.
We expect that the difference $\epsilon$ will shrink as the links are iteratively removed until the obtained graph is close to the graph $G$ whose effective resistance matrix $\Omega_G$ exactly equals the demand matrix $D$.  
To determine which link should be removed, in line 9, we compute the $N\times N$ heuristic score matrix
\[
\Gamma \gets \left(\hat{\Omega} - \widetilde{A}\right) \circ \left(D - \Omega \right) \circ A,
\]
where the $N\times N$ matrix $\hat{\Omega}$ contains the element-wise reciprocal of the effective resistance on the off-diagonal elements, with the diagonal elements remaining zero.
The heuristic score consists of three components:
\begin{enumerate}
    \item The first factor $\left(\hat{\Omega} - \widetilde{A}\right)$ of the heuristic score matrix $\Gamma$ denotes the reciprocal of effective resistance between a pair of adjacent nodes (i.e. $a_{ij}=1$), in case the direct link between them is removed (as in (\ref{Eq:Omega_ij_parallel})).
    Specifically, from (\ref{Eq:Omega_ij_parallel}), we have 
    \[\frac{1}{\left(\omega_{G\setminus l}\right)_{ij}} = \frac{1}{\omega_{ij}} - \frac{1}{r_{ij}} \]
    A smaller $\left(\omega_{G\setminus l}\right)_{ij}$, equivalently a larger $\frac{1}{\left(\omega_{G\setminus l}\right)_{ij}}$, may imply a smaller contribution of $\frac{1}{r_{ij}}$ and thus a larger $r_{ij}$, indicating that link $l = i \sim j$ can be regarded as redundant.
    Moreover, a smaller $\left(\omega_{G\setminus l}\right)_{ij}$ indicates that, after removing link $l = i \sim j$, nodes $i$ and $j$ may remain connected through multiple paths acting in parallel.  
    In other words, the adjacent nodes $i$ and $j$ are easily reachable via the rest of the graph when the link is removed.
    \item The second factor $\left(D-\Omega\right)$ quantifies the gap between the predetermined demand and the obtained effective resistance in the current iteration.
    \item The third factor is the adjacency matrix, which encodes the graph topology at the current iteration and ensures that only existing links are considered for removal.
\end{enumerate}
By jointly accounting for the link redundancy, the reachability between two nodes after link removal and the difference between the current effective resistance and the target demand, we remove the existing link with the highest value in $\Gamma$ (lines 10 and 11).
The adjacency matrix $A$, the weighted adjacency matrix $\widetilde{A}$ of the flow network, the corresponding effective resistance matrix $\Omega$ and the difference $\epsilon$ between the demand and the effective resistance matrix are then updated respectively in lines 12-14.
We perform the link removal until the updated $\epsilon$ is larger than the difference $\epsilon^\prime$ in the previous iteration.
At that point, the last removed link is returned (line 16).

In lines 17-19, we further scale the link weight matrix $W$ according to the norm $\| D-\Omega\|$.
Since the effective resistance $\Omega$ is sensitive even to insignificant perturbations~\cite {qiu2023inverse}, we cannot independently scale each element of the effective resistance $\Omega$ with a separate parameter.
Fortunately, scaling the entire effective resistance matrix $\Omega$, by a positive scaling parameter $\alpha$ yielding $\alpha\Omega$, guarantees that the corresponding weighted adjacency matrix remains nonnegative: From \eqref{Eq:omega_ij}, we have
\begin{align}\label{Eq:alphaomega_ij}
    \alpha\omega_{ij} &= \alpha\left(e_i - e_j\right)^{T}  Q^{\dagger} \left(e_i - e_j\right)\\
    &= \left(e_i - e_j\right)^{T}  \left(\frac{1}{\alpha}Q\right)^{\dagger} \left(e_i - e_j\right),
\end{align}
which indicates that $\alpha\Omega$ is the effective resistance of a graph with Laplacian $\frac{1}{\alpha}\widetilde{Q}$ and weighted adjacency matrix $\frac{1}{\alpha}\widetilde{A}$.
The remaining task is then to obtain the scaling parameter $\alpha$ that minimizes the norm $\| D-\Omega\|$, e.g., $L1$ norm, $\| D-\alpha\Omega\| = \sum_{ij} |d_{ij}-\alpha\omega_{ij}|$.
Because both the demand $d_{ij}$ and effective resistance $\omega_{ij}$ are positive, each term $|d_{ij}-\alpha\omega_{ij}|$ can be expressed piecewise as
\[
|d_{ij} - \alpha \omega_{ij}| =
\begin{cases}
d_{ij} - \alpha \omega_{ij}, & \text{if } \alpha \le \dfrac{d_{ij}}{\omega_{ij}}, \\[6pt]
\alpha \omega_{ij} - d_{ij}, & \text{if } \alpha > \dfrac{d_{ij}}{\omega_{ij}},
\end{cases}
\]
which is a convex piecewise-linear function, minimized at $\alpha = \dfrac{d_{ij}}{\omega_{ij}}$.
Hence, the norm $\| D-\alpha\Omega\| = \sum_{ij} |d_{ij}-\omega_{ij}|$ is also a convex piecewise linear function and the minimum occurs at one of the breakpoints, corresponding to a $\dfrac{d_{ij}}{\omega_{ij}}$ for a (or some) node pair $(i,j)$.
In line 18, instead of searching over all possible $\dfrac{d_{ij}}{\omega_{ij}}$ to obtain the minimum of norm $\| D-\alpha\Omega\|$, we use $\alpha$ equaling the mean of $\dfrac{d_{ij}}{\omega_{ij}}$ for all possible $(i,j)$, which can reduce the computational complexity.
In line 19, the weighted adjacency matrix $\widetilde{A}$ is updated by scaling the link weight matrix $W$ with the factor $\frac{1}{\alpha}$.
Finally, in line 20, the RGP returns the resulting weighted adjacency matrix $\widetilde{A}$.

The main computational complexity of the RGP algorithm comes from updating the effective resistance (line 13), which requires computational complexity $O(N^3)$ in each iteration.
Since the RGP algorithm initialises the topology with a complete graph and iteratively removes redundant links until, in the worst case, the graph becomes a tree graph with $N-1$ links, there are up to $\frac{N \left(N - 1 \right)}{2} - (N - 1)$ iterations.
Therefore, the overall computational complexity for the worst case of the RGP algorithm is $O(N^5)$.
To streamline the computational complexity, one possible strategy is to employ the Sherman-Morrison equation \cite{van1996matrix}, which provides an efficient way to compute the inverse of a matrix after a rank-one modification.
In each iteration of RGP, we remove one link from the graph, which corresponds to a rank-one update to the Laplacian $Q$.
Hence, after the first iteration, the Sherman-Morrison equation can eliminate the need to compute the pseudoinverse of the Laplacian $Q^\dagger$ and effective resistance $\Omega$ from scratch, which reduces the computational complexity of updating $\Omega$ from $O(N^3)$ to $O(N^2)$.
Accordingly, the overall computational complexity of the RGP can be reduced to $O(N^4)$.
% Another possible approach is to remove more than $1$ links in each iteration, especially in the early stage of removing links from a complete graph.
% Specifically, instead of obtaining the maximum element in $\Gamma$, we can remove redundant links with the maximum $m$ elements in $\Gamma$ in one iteration.
% The overall computational complexity can then shrink to $O(N^4)$ if we remove $m = N$ elements in one iteration.

In Appendix~\ref{app:Performance of RGP algorithm}, we comprehensively evaluated the performance of the RGP algorithm.
We introduce a benchmark approach, \quotes{Fiedler approach}, which executes \eqref{equation:getAtilde} by substituting the effective resistance matrix $\Omega$ with an input demand matrix $D$, as a reference.
The graphs produced by both the RGP algorithm and the Fiedler approach are then compared under the same predetermined demand matrix. 

As shown in Figs.~\ref{fig:performance_tree},~\ref{fig:performanceRGP},~\ref{fig:performance_ERexp} and Table~\ref{tab:RGP empirical network} in Appendix~\ref{app:Performance of RGP algorithm}, the RGP algorithm generates graphs whose effective resistance matrices closely approximate the predetermined demand matrix and demonstrates consistent performance across different demands.
Nearly all links in the graph produced by RGP also appear in the graph generated by the benchmark method. 
Compared with the benchmark approach, RGP generally produces sparser graphs.
Specifically, when the demand matrix equals an effective resistance derived from a tree graph, the RGP algorithm can generally generate a tree identical to that produced by the benchmark, Fig.~\ref{fig:performance_tree}.

\section{Conclusion}\label{sec:Conclusion}
Motivated by practical scenarios in 6G communication and information spreading on social platforms, where transport cannot be accurately described by a single shortest path, we study flow networks in which transport distributes proportionally across all available paths. 
We first analyze the size of the flow subgraph in flow networks. 
We decompose the giant component (GC) into a backbone subgraph $\mathcal{B}$ and branches.
Using the requirement that every non-terminal node in the flow subgraph must have at least two neighbors within the flow subgraph, we establish a self-consistent formulation. 
From this formulation, we obtain the probability \(p_b\) that a node has at least one neighbor that belongs to the backbone and the backbone fraction \(b\). Both probability \(p_b\) and backbone fraction \(b\) remain zero when the expected degree $E[D]<1$ and become positive once the expected degree exceeds one. Based on these quantities, we characterize the expected size of the flow subgraph. Conditioned on both nodes lying in the giant component, the normalized size of the flow subgraph converges to \(b\,p_b^2\) with an \(O(1/N)\) finite-size deviation. A similar reasoning applies to the number of links in the flow subgraph, which converges to \(p_b^4\) up to the same correction order. Simulations agree with these results and reproduce both the large-scale behavior and the finite-size deviations.

% We further study the power dissipation during the transport at both the link and network levels.
% Specifically, for unweighted ER graphs or ER graphs with uniformly distributed links, the power dissipated on a link exhibits a power law decay with respect to the size of the graph $N$ and link density $p$.

Finally, we consider the challenge of constructing a graph with predetermined end-to-end power dissipation, which can be reformulated as the \emph{inverse effective resistance problem} that asks for a graph such that the corresponding effective resistance matrix equals a given demand matrix.
Inspired by circuit rules and effective resistance computation, we propose a heuristic algorithm, \quotes{Resistor Gap Pruning} (RGP), which provides sparse graphs closely approximating the demand effective resistance and shows stable performance across different demand scenarios.
Further, our simulations reveal that nearly all links belonging to the sparse graphs $H$ obtained by RGP also exist in the benchmark ER graph $G$ by Fiedler approach, indicating that the links in graph $H$ exhibit significant influence on the effective resistance.
Indeed, for an ER graph with identical link weights, previous studies ~\cite{von2014hitting, sylvester2021random, akara-pipattana_resistance_2022, carmi2007transport,carmi2008transport} have shown that the effective resistance can be approximated as $\omega_{st} \propto \frac{1}{d_s} + \frac{1}{d_t}$, stating that the transport is dominated by the links adjacent to node $s$ and $t$, while the remaining graph is practically a perfect conductor.
Our findings then provide a complementary perspective on those \quotes{critical} links/nodes that determine effective resistance in a flow network: the effective resistance can be dominated by a small number of links that interconnect the network.
Moreover, the proposed RGP algorithm can also be utilized for a type of \quotes{network sparsification}, which provides a sparser network while preserving a similar effective resistance matrix. 

\section*{Author Contributions}

Z. Qiu performed all simulations, initiated and addressed the research problem studied in Section~IV, and led the writing of Sections~I (Introduction), II (Terminology), IV, and V (Conclusion).
X. Liu developed the analytical framework for computing the expected size of the flow subgraph, wrote Section~III, and proved and wrote Appendix~B (Lemmas for flow subgraph in graphs with i.i.d.\ continuous link weights).
R. Noldus introduced the concepts of the flow subgraph and end-to-end power dissipation, and reviewed the manuscript.
P. Van Mieghem posed the problem of the size of the flow subgraph, acted as the principal advisor for the project, and reviewed the manuscript.

\begin{acknowledgments}
P. Van Mieghem is supported by the European Research Council under the European Union’s Horizon 2020 research and innovation program (Grant Agreement 101019718).
\end{acknowledgments}

\newpage
\appendix
\section{Symbol}
\label{app:Notation list}
\setlength{\arrayrulewidth}{0.6pt}
\begin{table}[H]
	\centering
	\caption{Symbol}
	\label{symbol}
	\begin{tabular}{ll}
		\hline
		Symbol& Definition\\
		\hline
		$G$& A network or graph\\
		$\mathcal{N}$& Set of nodes\\
		$N$& Number of nodes\\
		$\mathcal{L}$& Set of links\\
		$L$& Number of links\\
        $l = i\sim j$& a link connecting node i and j\\
		$\mathcal{P}_{ij}$& Path from node $i$ to node $j$\\
		$A$& Adjacency matrix\\
		$W$& Link weight matrix\\
		$\widetilde{A}$& Weighted adjacency matrix\\
        $B$& Incidence matrix\\
		% $\widetilde{A}_F$& Weighted adjacency matrix of a flow network\\
        $Q$ & Laplacian matrix\\
        $\widetilde{Q}$ & Weighted Laplacian matrix\\
		$Q^\dag$& Pseudoinverse Laplacian matrix\\
		$\Omega$& Effective resistance matrix\\
        $R_G$ & Effective graph resistance\\
		$r_l$& Resistance of link $l$\\
		$S$& Shortest path weight matrix\\
        $d_i$& Degree of node $i$\\
		$\widetilde{\Delta}$& Weighted degree matrix\\
		% $\widetilde{\Delta}_F$& Weighted degree matrix of a flow network\\
		$u$ & All-one vector\\
		$J$ & All-one matrix\\
		$I$ & Identity matrix\\
		$e_i$& $N\times 1$ basic vector has only one\\ &  non-zero element $(e_i)_i = 1$\\
        $x_i$& External current injected into node $i$.\\
        $v_i$& Potential or voltage potential of node $i$\\
        $y_l$& Current through link $l$\\
        % $y_l(s,t)$& Current through link $l$ with node $s$ as source\\ &  and $t$ as destination\\
        $P_G$& Power dissipation in a network $G$\\
        $G_{ij}^*$& Flow subgraph with $(i,j)$ as source-destination\\ &  node pair\\
        $\mathcal{B}$ & Backbone subgraph in the flow subgraph\\
        $T_k$ & k-th brach in the flow subgraph\\
        $\rho_N$& Fraction of node number of the flow subgraph\\ &  normalized by node number of the graph\\
        $\rho_L$& Fraction of link number of the flow subgraph\\ &  normalized by link number of the graph\\
		\hline
	\end{tabular}
\end{table}

% \newpage
% \section{Expected link power dissipation in weighted ER networks}\label{app:Expected link power dissipation in weighted ER networks}
% \begin{figure}[!htbp]
%     \centering

%     % ---------- 子图 (a) ----------
%     \begin{subfigure}[t]{0.45\textwidth}
%         \includegraphics[width=\linewidth]{figure/powerdissipation/link_power_p0.10_withdiffN_cropped.pdf}
%         \caption{Weighted ER graphs $G_{0.1}(N)$}
%         \label{sfig:linkpowerwithdiffNweightedER}
%     \end{subfigure}
%     \hfill
%     % ---------- 子图 (b) ----------
%     \begin{subfigure}[t]{0.45\textwidth}
%         \includegraphics[width=\linewidth]{figure/powerdissipation/link_power_N200_withdiffp_cropped.pdf}
%         \caption{Weighted ER graphs $G_p(200)$}
%         \label{sfig:linkpowerwithdiffpweightedER}
%     \end{subfigure}

%     % ---------- 总标题 ----------
%     \caption{(a) The expected link power dissipation as a function of the size of the graph $N$ for ER networks $G_{p=0.1}(N)$.
%     (b) The expected link power dissipation as a function of the probability $p$ for ER networks $G_p(N=200)$.
%     For each simulation, we generate an ER graph $G_p(N)$ with link weights uniformly distributed in $(0,1)$.
%     We compute the link power dissipation for all possible node pairs and obtain the expected value $E[\Lambda_l]$.
%     The simulation is repeated $100$ times for $N<1000$ and once for $N \geq 1000$.}
%     \label{fig:linkPowerDissipationERweightedER}
% \end{figure}

\section{Lemmas for flow subgraph in graphs with i.i.d. continuous link weights}
\label{app:B}

\begin{lemma}
\label{lem:weighted_prob0_equipotential}
Consider the resistor network model in Section~\ref{sec:Electrical resistor network},
where the link resistances $\{r_l : l\in\mathcal L\}$ are mutually independent
continuous random variables.
Fix a source--destination pair $(i,j)$ and inject a unit current.

Let $s\neq t$ be two nodes connected by a path that contains at least one link
with nonzero current. Then
\[
\Pr[v_s = v_t] = 0 .
\]
\end{lemma}

\begin{proof}
For each link $l\in\mathcal L$, Ohm's law gives the voltage drop $U_l$ on link $l$ as
\[
U_l = r_l\, y_l .
\]
Fix a path $\mathcal{P}_{st}$ from node $s$ to node $t$.
By Kirchhoff's voltage law, the nodal potential difference between nodes $s$ and $t$ equals the signed sum of the voltage drops along the path $\mathcal{P}_{st}$,
\[
v_s - v_t
= \sum_{l\in \mathcal{P}_{st}}  U_l
= \sum_{l\in \mathcal{P}_{st}}  y_l\, r_l .
\]
Consequently, the equality of the nodal potentials at $s$ and $t$ is equivalent to the condition
\[
v_s = v_t
\quad\Longleftrightarrow\quad
\sum_{l\in \mathcal{P}_{st}} y_l\, r_l = 0 .
\]
Since the path $\mathcal{P}_{st}$ contains at least one link $l_0$ with nonzero current $y_{l_0}\neq 0$, the coefficients
$\{ y_l : l\in \mathcal{P}_{st}\}$ are not all zero.
Define the scalar random variable
\[
S := \sum_{l\in \mathcal{P}_{st}} y_l\, r_l .
\]
Conditioning on the resistances $\{r_l : l\in \mathcal{P}_{st},\, l\neq l_0\}$, the random variable $S$ is an affine function of the independent continuously distributed resistance $r_{l_0}$.
Therefore,
\[
\Pr\!\left[S = 0 \mid \{r_l : l\in \mathcal{P}_{st},\, l\neq l_0\}\right] = 0,
\]
because a continuously distributed random variable attains any prescribed value with probability zero
(see, e.g., \cite[Sec.~5.3]{ross2014first}).
This yields $\Pr[v_s = v_t] = 0$.
\end{proof}

\begin{lemma}
% \begin{lemma}[Property~\ref{property1} implies flow-subgraph membership in weighted graphs]
\label{lem:prop1_weighted_node_in_FS}
Consider the resistor network model in Section~\ref{sec:Electrical resistor network}.
Assume that the link resistances $\{r_l:l\in\mathcal L\}$ are mutually independent random variables with continuous distributions.
Fix a source--destination pair $(i,j)$ and consider the associated flow subgraph $G^*_{ij}$.
Let $w\neq i,j$ be a node satisfying Property~1, namely, there exist two distinct neighbors
$m\neq n$ such that the links $l_{wm}=w\sim m$ and $l_{wn}=w\sim n$ belong to $\mathcal L(G)$ and the nodes
$m,n$ belong to the flow subgraph node set $\mathcal N(G^*_{ij})$.
If the nodal potentials at $m$ and $n$ satisfy $v_m\neq v_n$, then the node $w$ belongs to the flow subgraph node set,
\[
w\in \mathcal N(G^*_{ij}) .
\]
\end{lemma}
\begin{proof}
Assume that the potentials at the two neighbors $m$ and $n$ of node $w$ satisfy $v_m\neq v_n$.
Suppose, for contradiction, that the node $w$ does not belong to the flow subgraph node set, namely,
$w\notin \mathcal N(G^*_{ij})$.
Then every link incident to $w$ carries zero current, in particular the link currents on
$l_{wm} = w \sim m$ and $l_{wn} = w \sim n$ satisfy
\[
y_{l_{wm}}=0
\quad\text{and}\quad
y_{l_{wn}}=0 .
\]
For the link $l_{wu}=w\sim m$ with resistance $r_{l_{wm}}>0$, Ohm's law $v_w-v_m=r_{l_{wm}}\,y_{l_{wm}}$ yields
\[
v_w=v_m .
\]
For the link $l_{wn}=w\sim n$ with resistance $r_{l_{wn}}>0$, Ohm's law $v_w-v_n=r_{l_{wn}}\,y_{l_{wn}}$ yields
\[
v_w=v_n .
\]
Hence $v_m=v_n$, which contradicts the assumption $v_m\neq v_n$.
Therefore, at least one of the two link currents $y_{l_{wm}}$ or $y_{l_{wn}}$ is nonzero and the node $w$ is incident to a nonzero-current link.
By the definition of the flow subgraph, this implies $w\in \mathcal N(G^*_{ij})$.
\end{proof}

\section{Performance of RGP algorithm}
\label{app:Performance of RGP algorithm}
\begin{figure*}[]
    \centering
    % ---------- 第一行 ----------
    \begin{subfigure}[t]{0.32\textwidth}
        \includegraphics[width=1\linewidth]{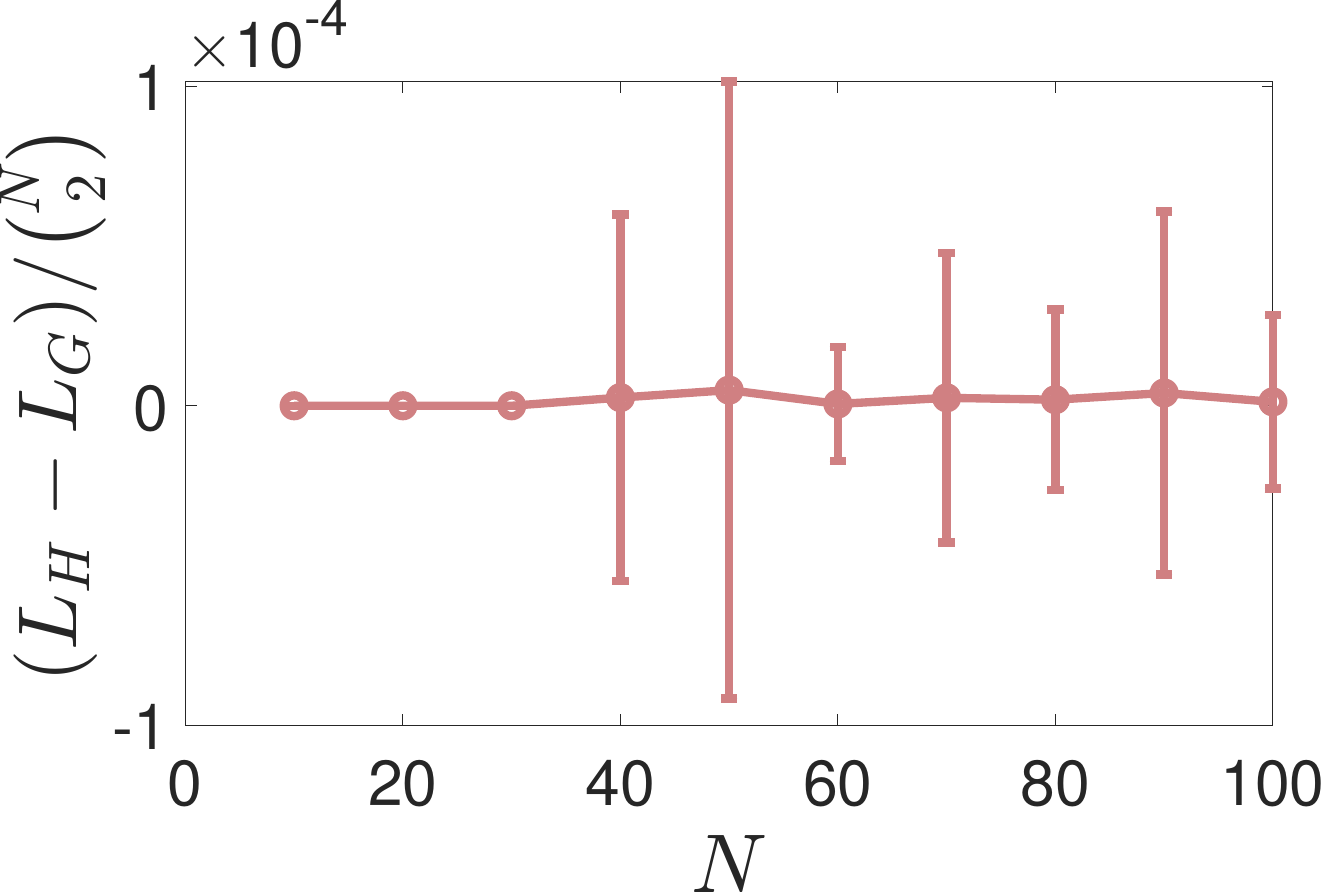}
        \caption{Normalized number of additional links}
        \label{sfig:treelinkremovalnum}
    \end{subfigure}
    \begin{subfigure}[t]{0.32\textwidth}
        \includegraphics[width=1\linewidth]{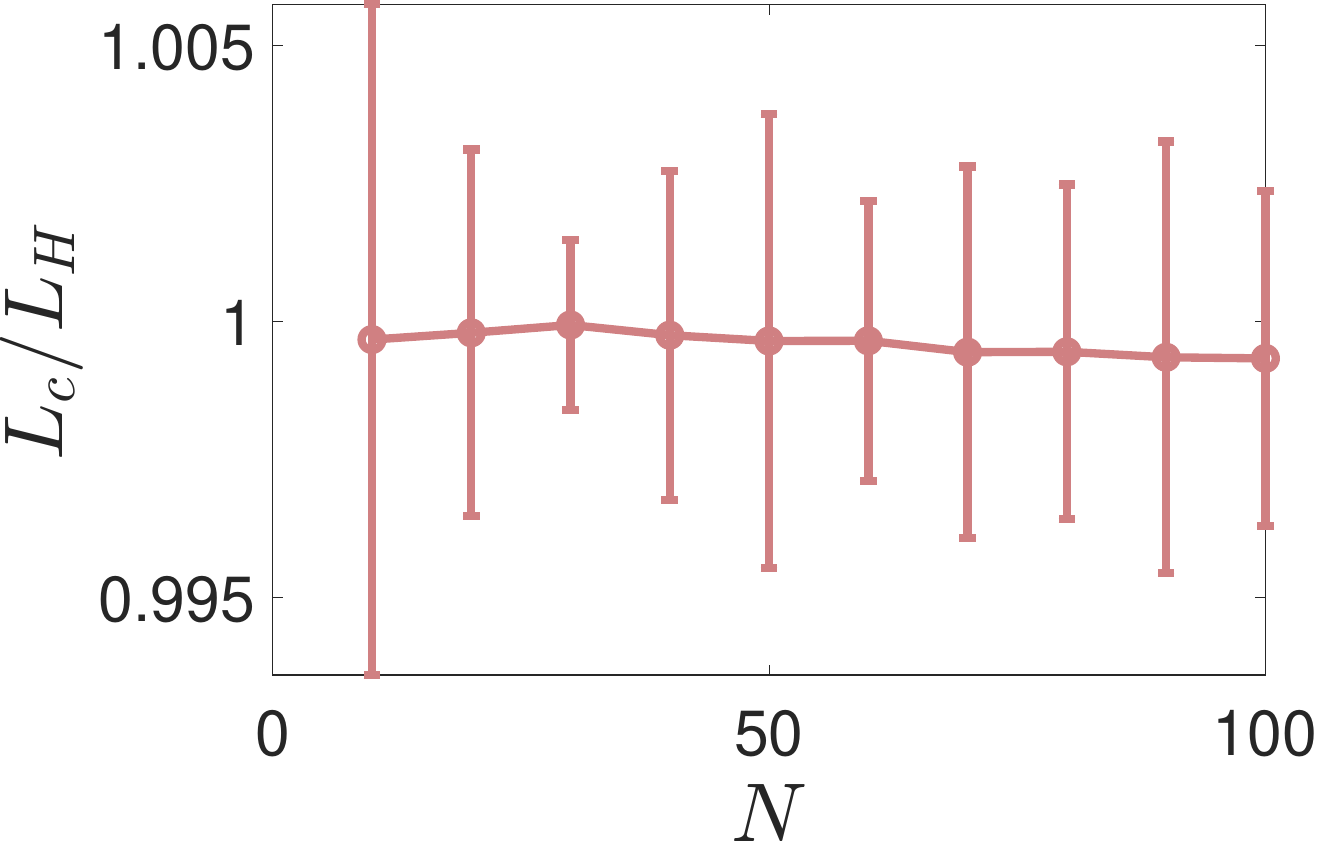}
        \caption{Normalized number of common links}
        \label{sfig:treecommonlink}
    \end{subfigure}
    % ---------- 第二行 ----------
    \begin{subfigure}[t]{0.32\textwidth}
        \includegraphics[width=1\linewidth]{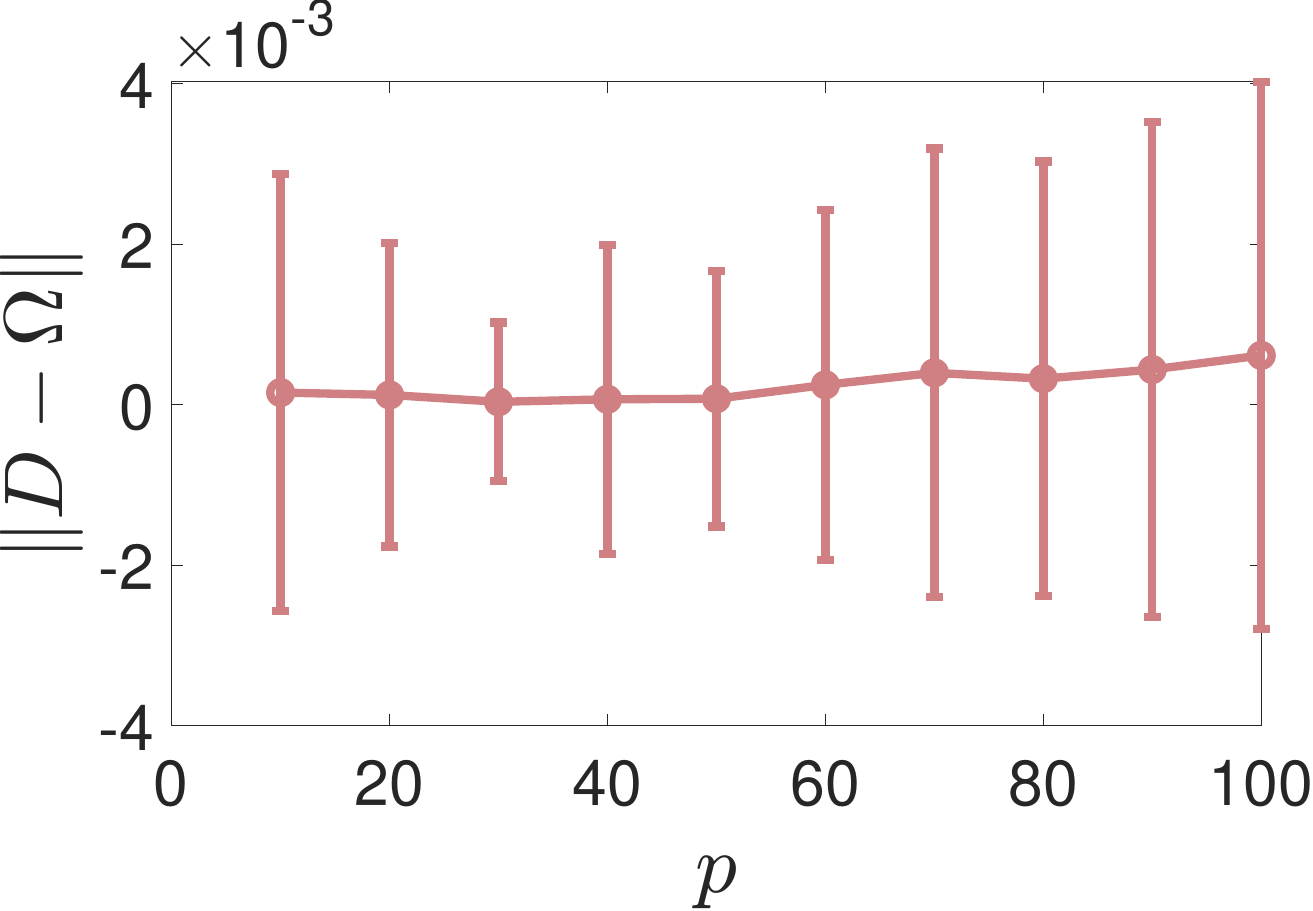}
        \caption{Norm difference}
        \label{sfig:treenorm}
    \end{subfigure}
    % ---------- 总标题 ----------
    \caption{Performance evaluation of RGP algorithm with tree baseline graphs. (a) Normalized number of additional links in the resulting graph $H$ compared with the baseline $G$.
    (b) Normalized number of common links in the baseline graph $G$ and the resulting graph $H$.
    (c) shows the norm between the demand matrix $D$ and the resulting effective resistance matrix $\Omega$. 
    The x-axis of all three panels represents the number of nodes $N$ of tree graphs. 
    For each graph size $N$, 1000 trees are generated as baseline graphs, with integer link weights independently and uniformly drawn from $[1,10]$.
    Error bars denote the standard deviation.}
    \label{fig:performance_tree}
\end{figure*}

\begin{figure*}[t]
    \centering
    \begin{subfigure}[t]{0.32\textwidth}
        \includegraphics[width=\linewidth]{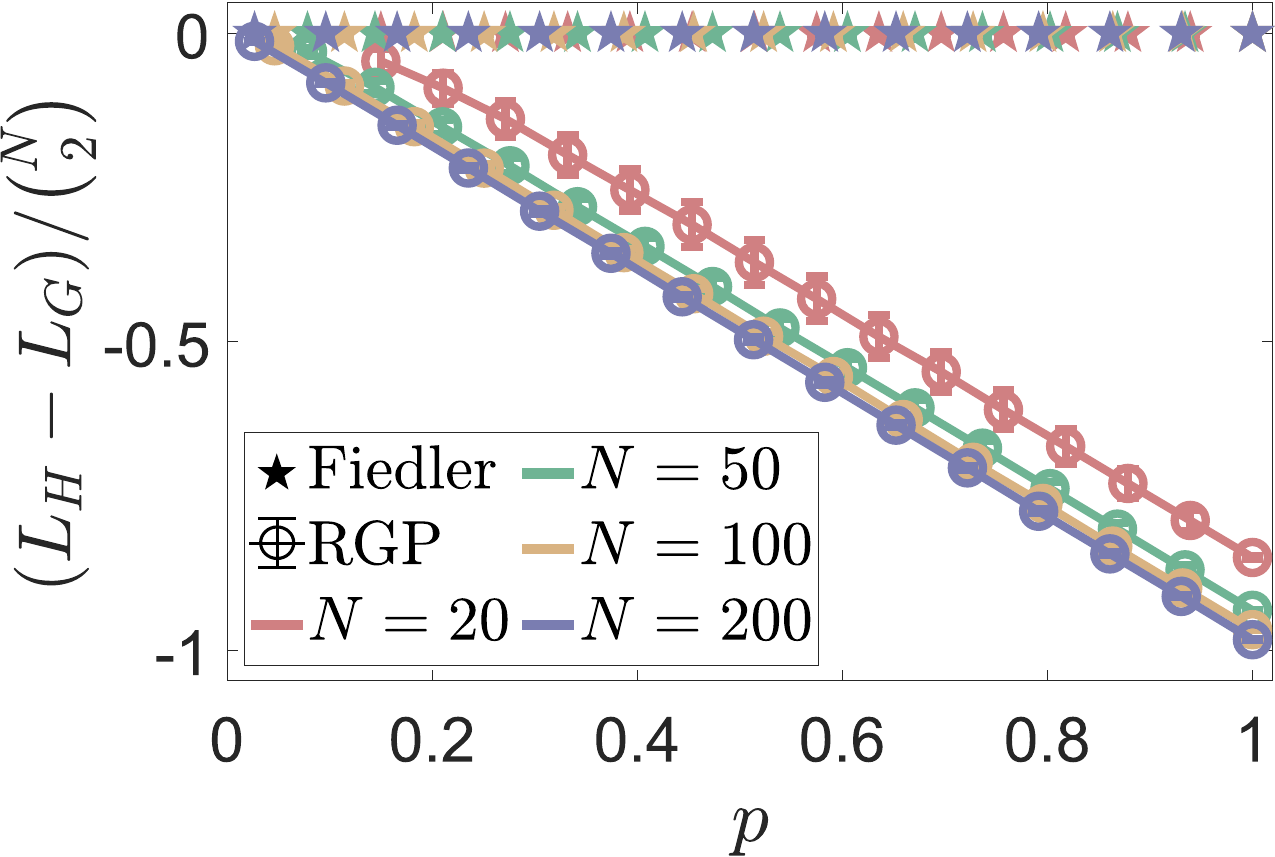}
        \caption{Normalized number of additional links}
        \label{sfig:ERlinkremovalnum}
    \end{subfigure}
    \begin{subfigure}[t]{0.32\textwidth}
        \includegraphics[width=\linewidth]{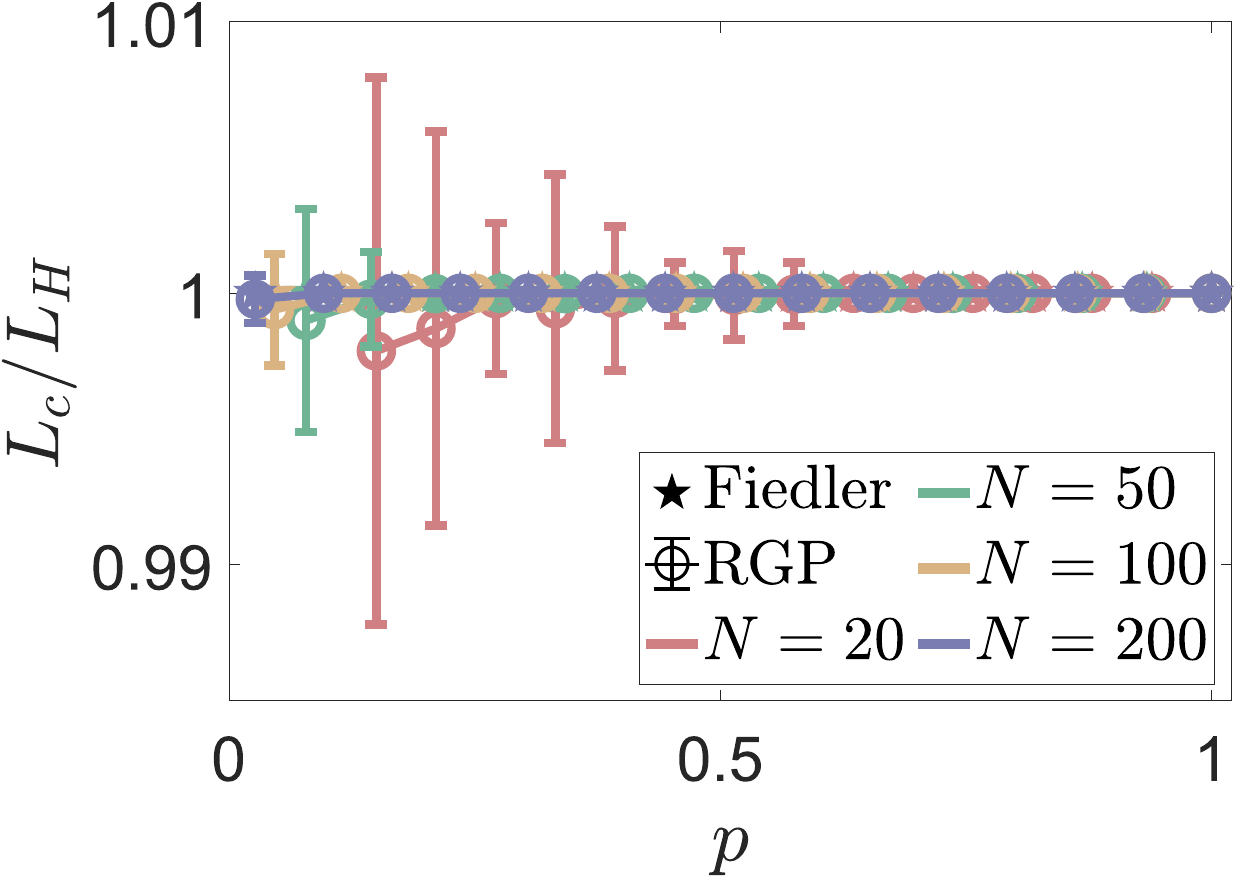}
        \caption{Normalized number of common links}
        \label{sfig:ERcommonlink}
    \end{subfigure}
    \begin{subfigure}[t]{0.32\textwidth}
        \includegraphics[width=\linewidth]{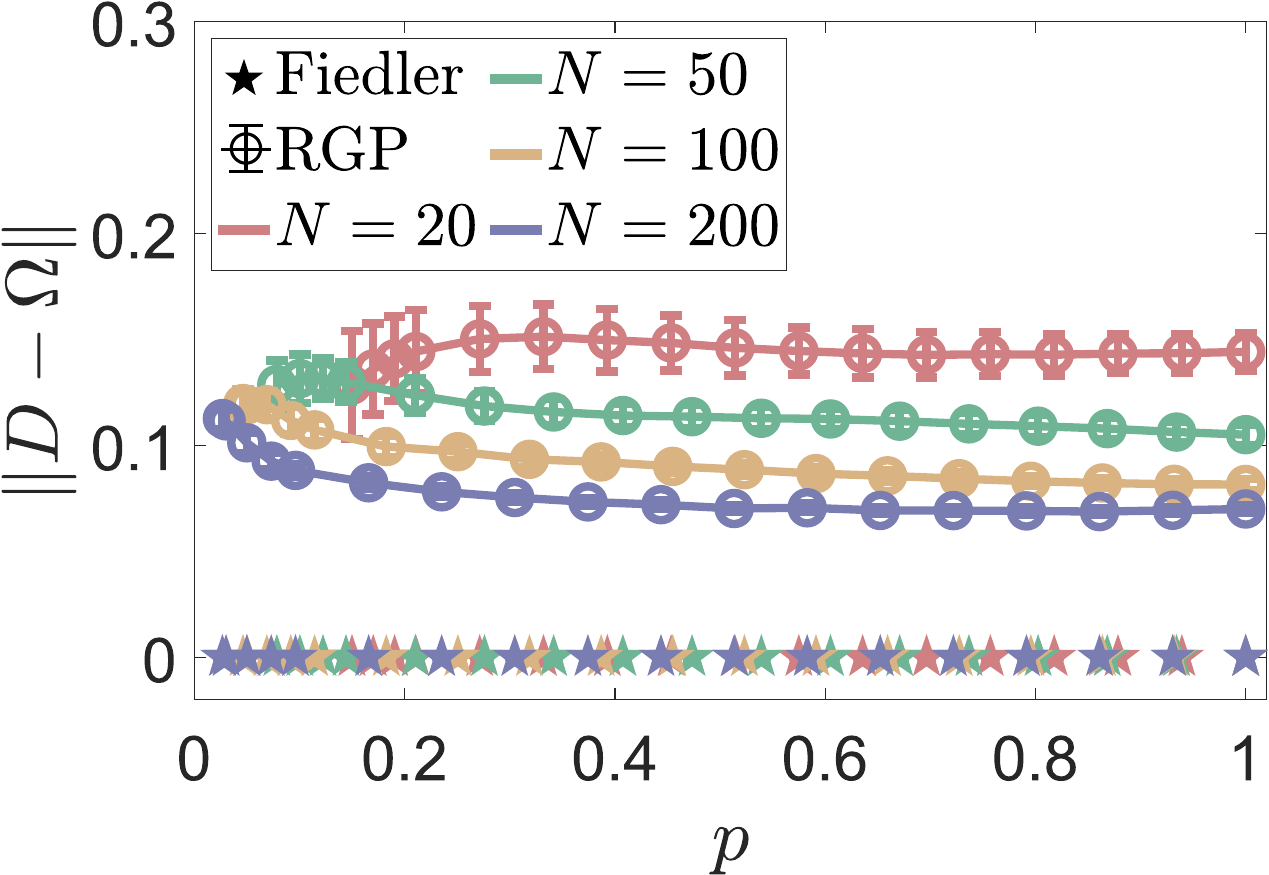}
        \caption{Norm difference}
        \label{sfig:ERnorm}
    \end{subfigure}
    \caption{Performance evaluation of the RGP algorithm (RGP) compared with the Fiedler approach (Fiedler).
    (a) Normalized number of additional links in the resulting graph $H$ compared with the baseline $G$.
    (b) Normalized number of common links in the baseline graph $G$ and the resulting graph $H$.
    (c) shows the norm between the demand matrix $D$ and the resulting effective resistance matrix $\Omega$. 
    The x-axis of all three panels represents the connection probability $p$ of ER graphs $G_p(N)$. 
    For each graph size $N$ and probability $p$, 1000 ER graphs with uniformly distributed link weights are generated as baseline graphs.
    Error bars indicate the standard deviation.}
    \label{fig:performanceRGP}
\end{figure*}

\begin{figure*}[!]
    \centering
    % ---------- 第一行 ----------
    \begin{subfigure}[t]{0.32\textwidth}
        \includegraphics[width=1\linewidth]{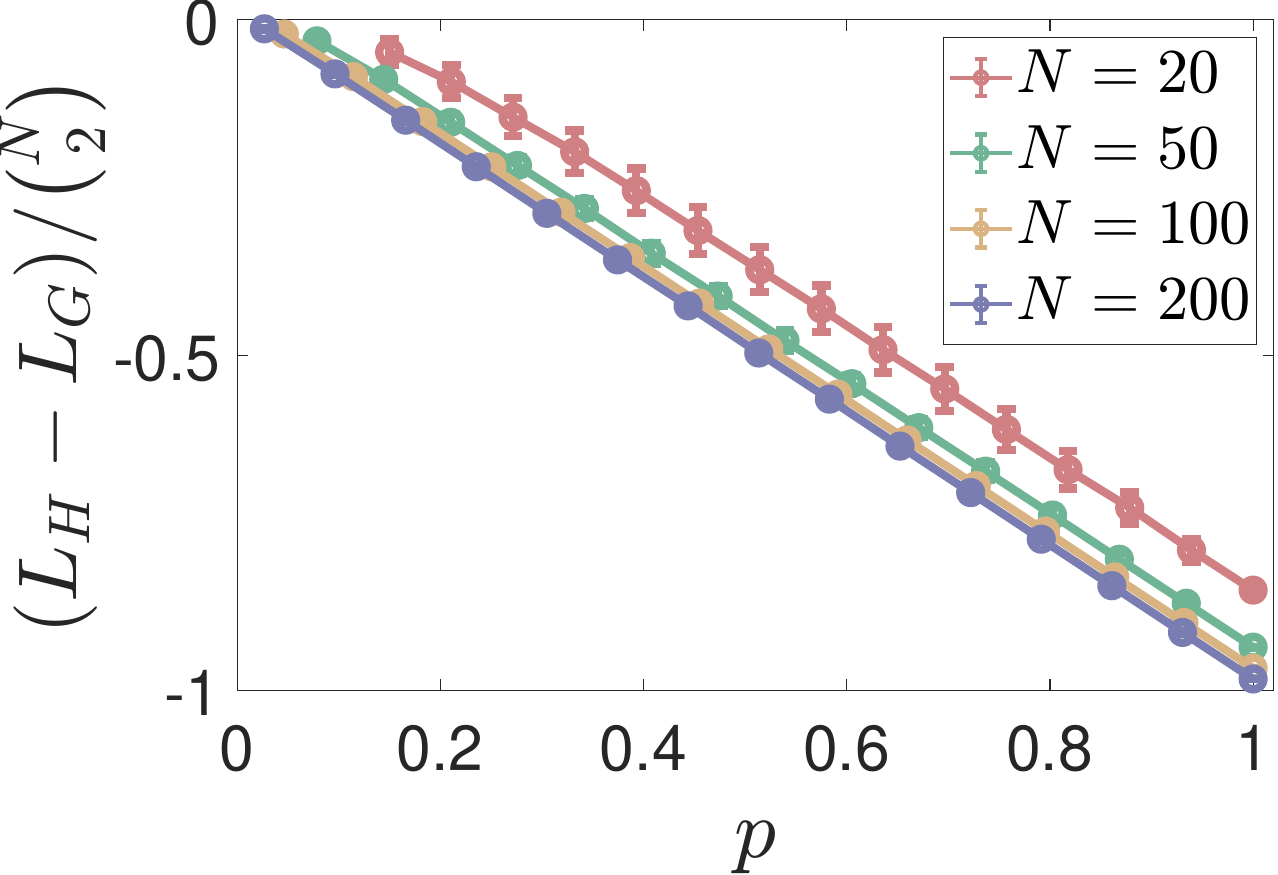}
        \caption{Normalized number of additional links}
        \label{sfig:ERlinkremovalnum_exp}
    \end{subfigure}
    \begin{subfigure}[t]{0.32\textwidth}
        \includegraphics[width=1\linewidth]{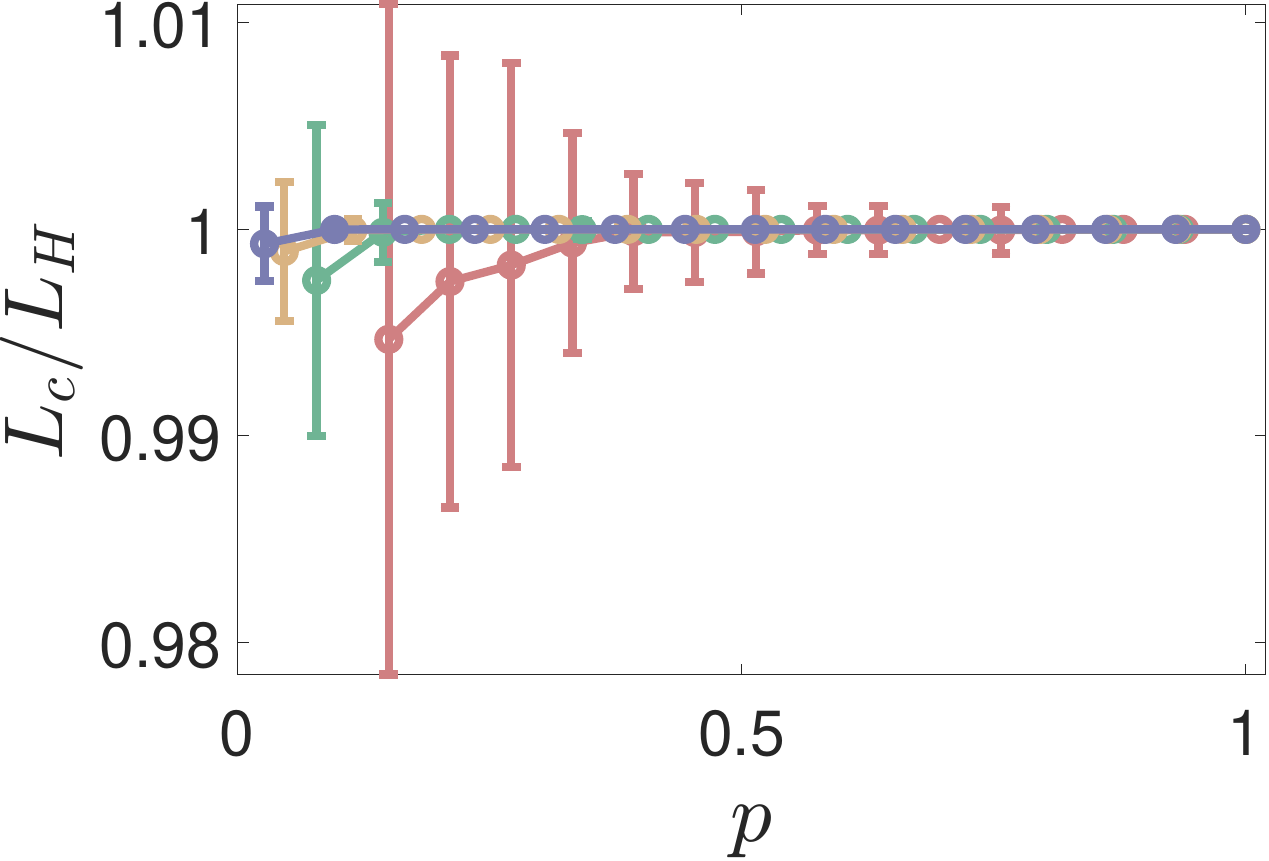}
        \caption{Normalized number of common links}
        \label{sfig:ERcommonlink_exp}
    \end{subfigure}
    % ---------- 第二行 ----------
    \begin{subfigure}[t]{0.32\textwidth}
        \includegraphics[width=1\linewidth]{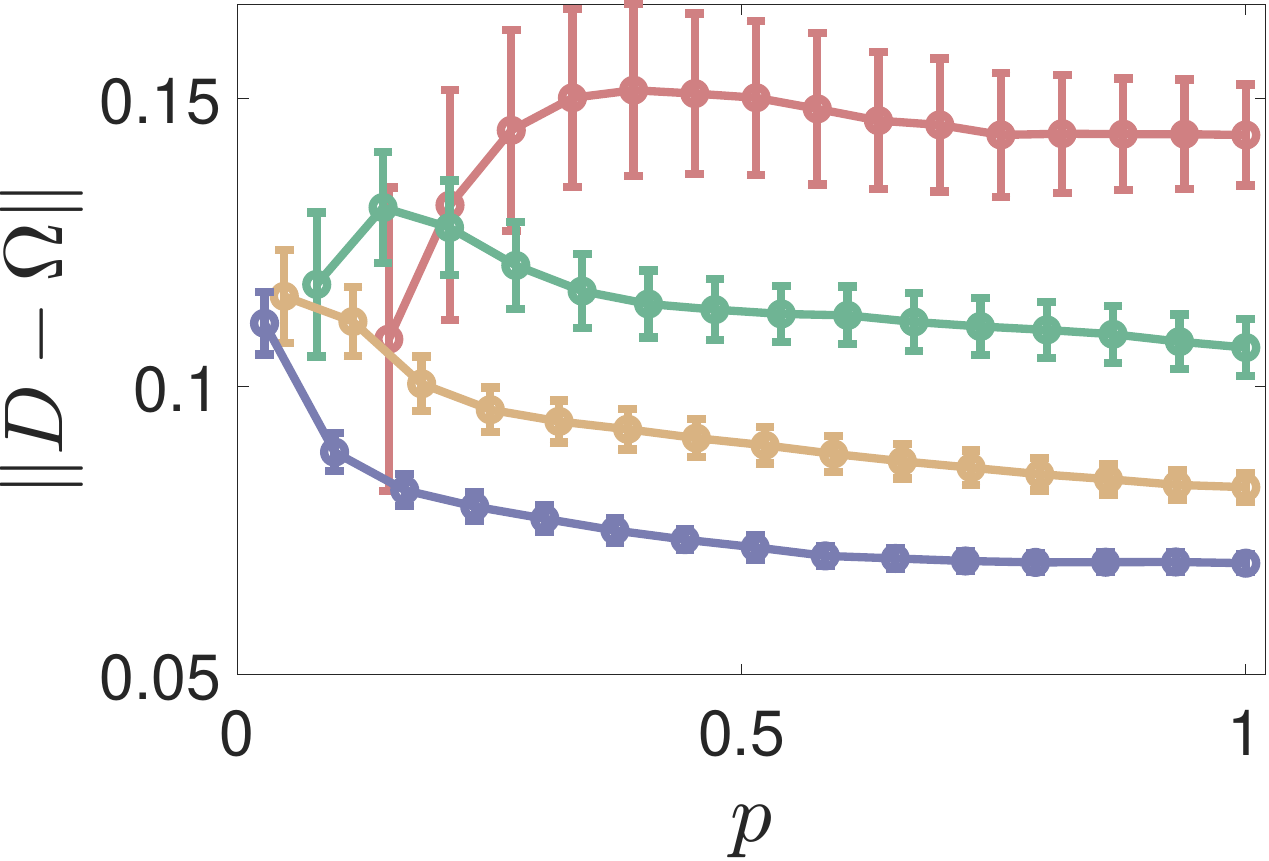}
        \caption{Norm difference}
        \label{sfig:ERnorm_exp}
    \end{subfigure}
    % ---------- 总标题 ----------
    \caption{(a) Normalized number of additional links in the resulting graph $H$ compared with the baseline $G$.
    (b) Normalized number of common links in the baseline graph $G$ and the resulting graph $H$.
    (c) Norm between the demand matrix $D$ and the resulting effective resistance matrix $\Omega$. 
    The x-axis of all three panels represents the connection probability $p$ of ER graphs $G_p(N)$. 
    For each graph size $N$ and probability $p$, 1000 ER graphs are generated as baseline graphs, each with link weights drawn from an exponential distribution with mean $0.5$.
    Error bars indicate the standard deviation.}
    \label{fig:performance_ERexp}
\end{figure*}

In this section, we evaluate the performance of the RGP algorithm.
We introduce a benchmark approach \quotes{Fiedler approach}, which executes \eqref{equation:getAtilde} by substituting the effective resistance matrix $\Omega$ with an input demand matrix $D$, as a reference.
Since the Fiedler approach is only feasible when the demand matrix is graph-realizable, our experiment is conducted given demand $D = \Omega_G$, where $\Omega_G$ is the effective resistance matrix of a randomly generated graph $G$.  
Specifically, for each experiment trial, we first generate a \quotes{baseline} graph $G$ and compute the corresponding effective resistance matrix $\Omega_G$.
With the demand matrix $D = \Omega_G$ as input, RGP and the Fiedler approach are respectively executed to obtain the resulting graph $H$ and the corresponding effective resistance matrix $\Omega$.
We then measure how the resulting graph $H$, produced by RGP, differs from the benchmark approach.
The performance of RGP is then assessed by three complementary criteria: (\romannumeral 1) the normalized number $\frac{2}{N(N-1)}(L_H-L_G)$ of additional links in the resulting graph $H$, (\romannumeral 2) the number of common links $L_c = \frac{1}{2} u^T \left(A\circ A_H\right) u$ shared by the baseline graph $G$ and the resulting graph $H$ and (\romannumeral 3) the norm $||D-\Omega||=\frac{1}{N(N-1)}\sum_{i}\sum_{j}\frac{|d_{ij}-\omega_{ij}|}{d_{ij}}$ of the demand matrix $D$ and the effective resistance matrix $\Omega$.

\begin{table}[h]
	\caption{Performance of RGP with empirical networks as baseline.}
	\label{tab:RGP empirical network}
	\centering
	\begin{tabular}{|c c c c c c|}
		\hline
		Network & $N$ & $L_G$ & $L_H-L_G$ & $L_c/L_H$ & $||D-S||$\\
		\hline
		karate   & 34&78&-35&1&0.1533\\
		dolphins   & 62&159&-79&1&0.1273\\
		NewSpainTravelMap   & 224&241&-13&1&0.0491\\
		wiki-vote   & 889&2914&-1745 &1&0.1045\\
		\hline
	\end{tabular}
\end{table}

We first consider the demand $D = \Omega_G$ derived from the sparsest graph, tree.
Given a tree graph $G$, the RGP algorithm can generally recover a graph $H = G$ with $D= \Omega_G$ as input, see Fig.~\ref{fig:performance_tree}.
In a tree graph, there exists a unique path $\mathcal{P}_{ij}$ between an arbitrary node pair $(i,j)$ and the effective resistance $\omega_{ij}$ is the sum of the resistances $r_l$ of links along the path $\mathcal{P}_{ij}$.
Starting from the initial complete graph (lines 1-3 in Algorithm~\ref{alg:RGP}) by RGP, whose link weights $W$ are the element-wise reciprocal of the demand $D$ over nonzero elements, any link $l = i \sim j$ that is absent in the baseline tree $G$ has a relatively high resistance $r_{l}$, as node $i$ and $j$ can already connect with each other via the path $\mathcal{P}_{ij}$ in the baseline tree $G$. 
Moreover, according to parallel series rules, the effective resistance $\omega_{ij}$ of a link $l = i \sim j$ that exists in the complete graph but not in the baseline tree $G$ is significantly smaller than the demand $d_{ij}$.
The RGP algorithm can then leverage both the high resistance $r_{ij}$ and large difference $d_{ij} - \omega_{ij}$ to remove redundant effective resistance links, ultimately recovering a tree graph $H$ that matches the given tree graph $G$.

We then extend the performance evaluation using demands $D = \Omega_G$ computed from ER graphs with different link densities $p$, Fig.~\ref{fig:performanceRGP}.
Specifically, for each size of the graph $N$ and link density $p$, 1000 ER graphs $G_p(N)$ with link weight uniformly distributed in $(0,1)$ are generated.
As shown in Fig.~\ref{fig:performanceRGP}, the Fiedler approach can exactly recover the given baseline graph $G_p(N)$, which is expected.
Compared to the benchmark Fiedler approach, RGP generally produces graphs $H$ with fewer links (Fig.~\ref{sfig:ERlinkremovalnum}).
The normalized number of additional links of resulting graphs $H$ by RGP decreases linearly as the link density $p$ increases. 
% The resulting graphs $H$ have a similar number of links $L_H$, irrespective of the number $L_G$ of links in the baseline graph $G$.
Fig.~\ref{sfig:ERcommonlink} presents the common links $L_c$ shared by graph $H$ and $G$ normalized by link number $L_H$.
The ratio $\frac{L_c}{L_H}$ is approximately $1$, indicating that nearly all the links in the reconstructed graph $H$ also exist in the given graph $G$.
In Fig.~\ref{sfig:ERnorm}, we further investigate the difference between the effective resistance matrix $\Omega$ of the resulting graph by RGP and the demand $D$.
% The norm $\| D-S \|$ exhibits a nonmonotonic behavior as a function of the link density $p$, increasing at very low link densities and decreasing as $p$ becomes larger for small size of the graph $N$.
The RGP algorithm performs better in graphs with a larger size $N$.

The feasibility of RGP is also examined for demands derived from ER graphs with different link weight distributions.
A similar pattern as Fig.~\ref{fig:performanceRGP} is visible in Fig.~\ref{fig:performance_ERexp}, where the link weights of the given ER graphs exhibit an exponential distribution with an average equaling $0.5$.
Further, in Table~\ref{tab:RGP empirical network}, we demonstrate the consistent performance of RGP across demands generated from different empirical networks~\cite{data, arzobispado_mexico_plano}.

% \section{Performance of RGP algorithm}
% \label{app:Performance of RGP algorithm}
% In this section, we supplement the performance of RGP in tree graphs, ER random graphs $G_p(N)$ with exponentially distributed link weights and four empirical networks~\cite{data, arzobispado_mexico_plano}.
% We use the demand matrix $D = \Omega_G$ as input, where $\Omega_G$ is the effective resistance of a randomly generated ER graph or an empirical network $G$.
% The performance of RGP is tested by (\romannumeral 1) the normalized number $\frac{2}{N(N-1)}(L_H-L_G)$ of additional links in the resulting graph $H$, (\romannumeral 2) the number of common links $L_c = \frac{1}{2} u^T \left(A\circ A_H\right) u$ shared by the baseline graph $G$ and the resulting graph $H$ and (\romannumeral 3) the norm $||D-\Omega||=\frac{1}{N(N-1)}\sum_{i}\sum_{j}\frac{|d_{ij}-\omega_{ij}|}{d_{ij}}$ of the demand matrix $D$ and the effective resistance matrix $\Omega$.
% Figs.~\ref{fig:performance_tree},~\ref{fig:performance_ERexp} and Table~\ref{tab:RGP empirical network} demonstrate the results for tree graphs, ER graphs and empirical networks, respectively.

\newpage
\bibliography{apssamp}% Produces the bibliography via BibTeX.

@PREAMBLE{
 "\providecommand{\noopsort}[1]{}" 
 # "\providecommand{\singleletter}[1]{#1}%" 
}

@article{GERNANDT202229,
title = {A Calderón type inverse problem for tree graphs},
journal = {Linear Algebra and its Applications},
volume = {646},
pages = {29-42},
year = {2022},
issn = {0024-3795},
doi = {https://doi.org/10.1016/j.laa.2022.03.018},
url = {https://www.sciencedirect.com/science/article/pii/S0024379522001197},
author = {Hannes Gernandt and Jonathan Rohleder}
}

@article{Alessandrini01011988,
author = {Giovanni Alessandrini},
title = {Stable determination of conductivity by boundary measurements},
journal = {Applicable Analysis},
volume = {27},
number = {1-3},
pages = {153--172},
year = {1988},
publisher = {Taylor \& Francis},
doi = {10.1080/00036818808839730},


URL = { 
    
        https://doi.org/10.1080/00036818808839730
    
    

},
eprint = { 
    
        https://doi.org/10.1080/00036818808839730
    
    

}

}

@inproceedings{data,
     title={The Network Data Repository with Interactive Graph Analytics and Visualization},
     author={Ryan A. Rossi and Nesreen K. Ahmed},
     booktitle={AAAI},
     year={2015}
}

@article{newman2003structure,
  author       = {Newman, M. E. J.},
  title        = {The Structure and Function of Complex Networks},
  journal      = {SIAM Review},
  volume       = {45},
  number       = {2},
  pages        = {167--256},
  year         = {2003},
  doi          = {10.1137/S003614450342480},
  url          = {https://doi.org/10.1137/S003614450342480}
}

@inproceedings{hassidim2013network,
  title={Network utilization: The flow view},
  author={Hassidim, Avinatan and Raz, Danny and Segalov, Michal and Shaqed, Ariel},
  booktitle={2013 Proceedings IEEE INFOCOM},
  pages={1429--1437},
  year={2013},
  organization={IEEE}
}

@article{aldous2008optimal,
  title={Optimal spatial transportation networks where link costs are sublinear in linkcapacity},
  author={Aldous, David J},
  journal={Journal of Statistical Mechanics: Theory and Experiment},
  volume={2008},
  number={03},
  pages={P03006},
  year={2008},
  publisher={IOP Publishing}
}

@incollection{holme2009diplomat,
  title={The diplomat’s dilemma: Maximal power for minimal effort in social networks},
  author={Holme, Petter and Ghoshal, Gourab},
  booktitle={Adaptive networks: Theory, models and applications},
  pages={269--288},
  year={2009},
  publisher={Springer}
}

@article{jackson2007study,
  title={The study of social networks in economics},
  author={Jackson, Matthew O},
  journal={The missing links: Formation and decay of economic networks},
  volume={76},
  pages={210--225},
  year={2007},
  publisher={Russell Sage Foundation, Thousand Oaks.[812] Jadbabaie, Ali, Pooya Molavi~…}
}

@article{guennebaud2024energy,
  title={Energy consumption of data transfer: Intensity indicators versus absolute estimates},
  author={Guennebaud, Ga{\"e}l and Bugeau, Aur{\'e}lie},
  journal={Journal of Industrial Ecology},
  volume={28},
  number={4},
  pages={996--1008},
  year={2024},
  publisher={Wiley Online Library}
}

@article{nowzari2016analysis,
  title={Analysis and control of epidemics: A survey of spreading processes on complex networks},
  author={Nowzari, Cameron and Preciado, Victor M and Pappas, George J},
  journal={IEEE Control Systems Magazine},
  volume={36},
  number={1},
  pages={26--46},
  year={2016},
  publisher={IEEE}
}

@article{pastor2015epidemic,
  title={Epidemic processes in complex networks},
  author={Pastor-Satorras, Romualdo and Castellano, Claudio and Van Mieghem, Piet and Vespignani, Alessandro},
  journal={Reviews of modern physics},
  volume={87},
  number={3},
  pages={925--979},
  year={2015},
  publisher={APS}
}

@article{qiu2022efficient,
  title={Efficient shortest path counting on large road networks},
  author={Qiu, Y.-X. and Wen, D. and Qin, L. and Li, W. and Li, R.-H. and Zhang, Y.},
  journal={Proceedings of the VLDB Endowment},
  year={2022},
  publisher={Association for Computing Machinery (ACM)}
}

@inproceedings{begtavseviu2001measurements,
  title={Measurements of the Hopcount in Internet},
  author={Begta{\v{s}}evi{\"u}, F and Van Mieghem, P},
  booktitle={Proceedings of Workshop on Passive and Active Measurement (PAM)},
  pages={183--190},
  year={2001}
}

@article{zhan1998shortest,
  title={Shortest path algorithms: an evaluation using real road networks},
  author={Zhan, F Benjamin and Noon, Charles E},
  journal={Transportation science},
  volume={32},
  number={1},
  pages={65--73},
  year={1998},
  publisher={INFORMS}
}

@article{PhysRevLett.85.5468,
  title = {Network Robustness and Fragility: Percolation on Random Graphs},
  author = {Callaway, Duncan S. and Newman, M. E. J. and Strogatz, Steven H. and Watts, Duncan J.},
  journal = {Phys. Rev. Lett.},
  volume = {85},
  issue = {25},
  pages = {5468--5471},
  numpages = {0},
  year = {2000},
  month = {Dec},
  publisher = {American Physical Society},
  doi = {10.1103/PhysRevLett.85.5468},
  url = {https://link.aps.org/doi/10.1103/PhysRevLett.85.5468}
}

@article{erd6s1960evolution,
  title={On the evolution of random graphs},
  author={Erd\H{o}s, P. and R{\'e}nyi, A.},
  journal={Publ. Math. Inst. Hungar. Acad. Sci},
  volume={5},
  pages={17--61},
  year={1960}
}

@article{PhysRevE.64.026118,
  title = {Random graphs with arbitrary degree distributions and their applications},
  author = {Newman, M. E. J. and Strogatz, S. H. and Watts, D. J.},
  journal = {Phys. Rev. E},
  volume = {64},
  issue = {2},
  pages = {026118},
  numpages = {17},
  year = {2001},
  month = {Jul},
  publisher = {American Physical Society},
  doi = {10.1103/PhysRevE.64.026118},
  url = {https://link.aps.org/doi/10.1103/PhysRevE.64.026118}
}

@book{newman2010networks,
    author = {Newman, M. E. J.},
    title = "{Networks: An Introduction}",
    publisher = {Oxford University Press},
    year = {2010},
    month = {03},
    isbn = {9780199206650},
    doi = {10.1093/acprof:oso/9780199206650.001.0001},
    url = {https://doi.org/10.1093/acprof:oso/9780199206650.001.0001}
}

@ARTICLE{8766143,
  author={Zhang, Zhengquan and Xiao, Yue and Ma, Zheng and Xiao, Ming and Ding, Zhiguo and Lei, Xianfu and Karagiannidis, George K. and Fan, Pingzhi},
  journal={IEEE Vehicular Technology Magazine}, 
  title={6G Wireless Networks: Vision, Requirements, Architecture, and Key Technologies}, 
  year={2019},
  volume={14},
  number={3},
  pages={28-41},
  doi={10.1109/MVT.2019.2921208}}

@Article{s22093136,
AUTHOR = {Dicandia, Francesco Alessio and Fonseca, Nelson J. G. and Bacco, Manlio and Mugnaini, Sara and Genovesi, Simone},
TITLE = {Space-Air-Ground Integrated 6G Wireless Communication Networks: A Review of Antenna Technologies and Application Scenarios},
JOURNAL = {Sensors},
VOLUME = {22},
YEAR = {2022},
NUMBER = {9},
ARTICLE-NUMBER = {3136},
URL = {https://www.mdpi.com/1424-8220/22/9/3136},
PubMedID = {35590826},
ISSN = {1424-8220},
DOI = {10.3390/s22093136}
}

@book{van1996matrix,
  title={Matrix computations},
  author={Golub, Gene H and Van Loan, Charles F},
  year={2013},
  publisher={JHU press}
}

@article{carmi2008transport,
  title={Transport in networks with multiple sources and sinks},
  author={Carmi, Shai and Wu, Zhenhua and Havlin, Shlomo and Stanley, H Eugene},
  journal={Europhysics Letters},
  volume={84},
  number={2},
  pages={28005},
  year={2008},
  publisher={IOP Publishing}
}

@article{carmi2007transport,
  title={Transport between multiple users in complex networks},
  author={Carmi, S and Wu, Z and L{\'o}pez, Emiliano and Havlin, Shlomo and Eugene Stanley, H},
  journal={The European Physical Journal B},
  volume={57},
  number={2},
  pages={165--174},
  year={2007},
  publisher={Springer}
}

@article{von2014hitting,
  title={Hitting and commute times in large random neighborhood graphs},
  author={Von Luxburg, Ulrike and Radl, Agnes and Hein, Matthias},
  journal={The Journal of Machine Learning Research},
  volume={15},
  number={1},
  pages={1751--1798},
  year={2014},
  publisher={JMLR. org}
}

@article{sylvester2021random,
  title={Random walk hitting times and effective resistance in sparsely connected Erd{\H{o}}s-R{\'e}nyi random graphs},
  author={Sylvester, John},
  journal={Journal of Graph Theory},
  volume={96},
  number={1},
  pages={44--84},
  year={2021},
  publisher={Wiley Online Library}
}

@article{bellman1958routing,
  title={On a routing problem},
  author={Bellman, Richard},
  journal={Quarterly of applied mathematics},
  volume={16},
  number={1},
  pages={87--90},
  year={1958}
}

@article{schrijver2012history,
  title={On the history of the shortest path problem},
  author={Schrijver, Alexander},
  journal={Documenta Mathematica},
  volume={17},
  number={1},
  pages={155--167},
  year={2012}
}

@book{cormen2022introduction,
  title={Introduction to algorithms},
  author={Cormen, Thomas H and Leiserson, Charles E and Rivest, Ronald L and Stein, Clifford},
  year={2022},
  publisher={MIT press}
}

@article{kitsak2023finding,
  title={Finding shortest and nearly shortest path nodes in large substantially incomplete networks by hyperbolic mapping},
  author={Kitsak, Maksim and Ganin, Alexander and Elmokashfi, Ahmed and Cui, Hongzhu and Eisenberg, Daniel A and Alderson, David L and Korkin, Dmitry and Linkov, Igor},
  journal={Nature Communications},
  volume={14},
  number={1},
  pages={186},
  year={2023},
  publisher={Nature Publishing Group UK London}
}

@article{gomathi2018energy,
  title={Energy efficient shortest path routing protocol for underwater acoustic wireless sensor network},
  author={Gomathi, RM and Martin Leo Manickam, J},
  journal={Wireless Personal Communications},
  volume={98},
  pages={843--856},
  year={2018},
  publisher={Springer}
}

@article{zhang2024mapreduce,
  title={A mapreduce-based approach for shortest path problem in road networks},
  author={Zhang, Dongbo and Shou, Yanfang and Xu, Jianmin},
  journal={Journal of Ambient Intelligence and Humanized Computing},
  pages={1--9},
  year={2024},
  publisher={Springer}
}

@book{fiedler2011matrices,
  title={{Matrices and graphs in geometry}},
  author={Fiedler, M.},
  year={2011},
  publisher={Cambridge University Press},
  address={Cambridge, U.K.}
}

@article{Mieghem2021atree,
  title={A tree realization of a distance matrix: the inverse shortest path problem with a demand matrix generated by a tree},
  author={Van Mieghem, P.},
  journal={Delft University of Technology, Report20211012},
  year={2021},
  pages = {1--15},
}

@article{van2017pseudoinverse,
  title={Pseudoinverse of the {Laplacian} and best spreader node in a network},
  author={Van Mieghem, P. and Devriendt, K. and Cetinay, H.},
  journal={Physical Review E},
  volume={96},
  number={3},
  pages={032311},
  year={2017},
  publisher={APS}
}

@article{qiu2023inverse,
  title={Inverse All Shortest Path Problem},
  author={Qiu, Zhihao and Joki{\'c}, Ivan and Tang, Siyu and Noldus, Rogier and Van Mieghem, Piet},
  journal={IEEE Transactions on Network Science and Engineering},
  year={2023},
  publisher={IEEE}
}

@book{van2014performance,
  title={Performance analysis of complex networks and systems},
  author={Van Mieghem, P.},
  year={2014},
  publisher={Cambridge University Press},
  address={Cambridge, U.K.}
}

@article{dijkstra1959note,
  title={A note on two problems in connexion with graphs},
  author={Dijkstra, E. W.},
  journal={Numerische mathematik},
  volume={1},
  number={1},
  pages={269--271},
  year={1959}
}

@inproceedings{fortz2000internet,
  title={Internet traffic engineering by optimizing {OSPF} weights},
  author={Fortz, B. and Thorup, M.},
  booktitle={Proceedings IEEE INFOCOM 2000},
  volume={2},
  pages={519--528},
  year={2000},
  organization={IEEE}
}

@article{fiedler1998some,
  title={Some characterizations of symmetric inverse {M}-matrices},
  author={Fiedler, M.},
  journal={Linear algebra and its applications},
  volume={275},
  pages={179--187},
  year={1998},
  publisher={Elsevier}
}

@article{van2010framework,
  title={A framework for computing topological network robustness},
  author={Van Mieghem, P. and Doerr, C. and Wang, H. and Hernandez, J. Martin and Hutchison, D. and Karaliopoulos, M. and Kooij, R.E.},
  journal={Delft University of Technology, Report20101218},
  pages={1--15},
  year={2010}
}

@book{PVM_GraphSpectra2023,
    title = {{Graph Spectra for Complex Networks}},
    year = {2023},
    author = {Van Mieghem, P.},
    publisher = {Cambridge University Press},
    edition = {Second},}

@article{BOZZO2013460,
title = {Resistance distance, closeness, and betweenness},
journal = {Social Networks},
volume = {35},
number = {3},
pages = {460-469},
year = {2013},
issn = {0378-8733},
doi = {https://doi.org/10.1016/j.socnet.2013.05.003},
url = {https://www.sciencedirect.com/science/article/pii/S0378873313000488},
author = {Enrico Bozzo and Massimo Franceschet}
}

@article{barabasi2013network,
  title={Network science},
  author={Barab{\'a}si, Albert-L{\'a}szl{\'o}},
  journal={Philosophical Transactions of the Royal Society A: Mathematical, Physical and Engineering Sciences},
  volume={371},
  number={1987},
  pages={20120375},
  year={2013},
  publisher={The Royal Society Publishing}
}

@article{erdos1963asymmetric,
  title={Asymmetric graphs},
  author={Erdos, Paul and R{\'e}nyi, Alfr{\'e}d},
  journal={Acta Math. Acad. Sci. Hungar},
  volume={14},
  number={295-315},
  pages={3},
  year={1963}
}

@article{NEWMAN200539,
title = {A measure of betweenness centrality based on random walks},
journal = {Social Networks},
volume = {27},
number = {1},
pages = {39-54},
year = {2005},
issn = {0378-8733},
doi = {https://doi.org/10.1016/j.socnet.2004.11.009},
url = {https://www.sciencedirect.com/science/article/pii/S0378873304000681},
author = {M.E. J. Newman}
}

@article{pearson1905problem,
  title={The problem of the random walk},
  author={Pearson, Karl},
  journal={Nature},
  volume={72},
  number={1867},
  pages={342--342},
  year={1905},
  publisher={Nature Publishing Group UK London}
}

@article{MASUDA20171,
title = {Random walks and diffusion on networks},
journal = {Physics Reports},
volume = {716-717},
pages = {1-58},
year = {2017},
note = {Random walks and diffusion on networks},
issn = {0370-1573},
doi = {https://doi.org/10.1016/j.physrep.2017.07.007},
url = {https://www.sciencedirect.com/science/article/pii/S0370157317302946},
author = {Naoki Masuda and Mason A. Porter and Renaud Lambiotte}
}

@article{Bozzo2012Approximations,
author = {Enrico Bozzo and Massimo Franceschet},
title = {{Approximations of the Generalized Inverse of the Graph Laplacian Matrix}},
volume = {8},
journal = {Internet Mathematics},
number = {4},
publisher = {A K Peters, Ltd.},
pages = {456 -- 481},
year = {2012},
}

@article{akara-pipattana_resistance_2022,
    title = {Resistance distance distribution in large sparse random graphs},
    volume = {2022},
    url = {https://doi.org/10.1088/1742-5468/ac57ba},
    doi = {10.1088/1742-5468/ac57ba},
    number = {3},
    journal = {Journal of Statistical Mechanics: Theory and Experiment},
    author = {Akara-pipattana, Pawat and Chotibut, Thiparat and Evnin, Oleg},
    month = mar,
    year = {2022},
    note = {Publisher: IOP Publishing and SISSA},
    pages = {033404},
}

@article{freeman_centrality_1991,
    title = {Centrality in valued graphs: {A} measure of betweenness based on network flow},
    volume = {13},
    issn = {0378-8733},
    url = {https://www.sciencedirect.com/science/article/pii/037887339190017N},
    doi = {https://doi.org/10.1016/0378-8733(91)90017-N},
    number = {2},
    journal = {Social Networks},
    author = {Freeman, Linton C. and Borgatti, Stephen P. and White, Douglas R.},
    year = {1991},
    pages = {141--154},
}

@book{ross2014first,
  title={Introduction to probability models},
  author={Ross, Sheldon M},
  year={2014},
  publisher={Academic press}
}

@book{van2024random,
  title={Random graphs and complex networks},
  author={Van Der Hofstad, Remco},
  volume={2},
  year={2024},
  publisher={Cambridge university press}
}

@article{carmona2024stable,
  title={Stable recovery of piecewise constant conductance on spider networks},
  author={Carmona, {\'A}ngeles and Encinas, Andr{\'e}s M. and Jim{\'e}nez, Mar{\'\i}a Jos{\'e} and Samperio, {\'A}lvaro},
  journal={International Journal of Computer Mathematics},
  volume={101},
  number={11},
  pages={1237--1254},
  year={2024},
  publisher={Taylor \& Francis},
  doi={10.1080/00207160.2024.2385631}
}

@misc{arzobispado_mexico_plano,
  author       = {{Anonymous}},
  title        = {Plano del Arzobispado de M{\'e}xico},
  institution  = {Instituto Nacional de Antropolog{\'\i}a e Historia},
  address      = {Mexico},
  year         = {2018},
  note         = {Accessed 9 June 2018}
}

@article{corless1996lambert,
  title={On the {Lambert} {W} function},
  author={Corless, Robert M. and Gonnet, Gaston H. and Hare, David E. G. and Jeffrey, David J. and Knuth, Donald E.},
  journal={Advances in Computational Mathematics},
  volume={5},
  number={1},
  pages={329--359},
  year={1996},
  publisher={Springer},
  doi={10.1007/BF02124750}
}

@ARTICLE{1717439,
  author={Chiasserini, C.-F. and Garetto, M.},
  journal={IEEE Transactions on Mobile Computing}, 
  title={An Analytical Model for Wireless Sensor Networks with Sleeping Nodes}, 
  year={2006},
  volume={5},
  number={12},
  pages={1706-1718},
  doi={10.1109/TMC.2006.175}}

\end{document}